\documentclass[11pt]{article}

\usepackage{hyperref}
\usepackage{fullpage}
\usepackage{natbib}

\usepackage{amsmath,amssymb,amsfonts,amsthm}
\usepackage{thm-restate}

\usepackage{url}            
\usepackage{booktabs}       
\usepackage{amsfonts}       
\usepackage{nicefrac}       
\usepackage{microtype}      
\usepackage{xcolor}         
\usepackage{xspace}
\usepackage[capitalize, noabbrev]{cleveref}
\usepackage{algorithm}
\usepackage{algorithmic}

\newcommand{\alglinelabel}{%
  \addtocounter{ALC@line}{-1}
  \refstepcounter{ALC@line}
  \label
}
\newcommand{\Comment}[1]{\hfill $\triangleright$\text{ #1}}

\def\final{0}  
\def\iflong{\iffalse}
\ifnum\final=0  
\newcommand{\tnote}[1]{{\color{teal}[\small {Tamalika: \bf #1}]}}

\newcommand{\todo}[1]{{\color{red}[\small {TODO: #1}]}}

\newcommand{\rc}[1]{{\color{blue}[\small {Rachel: \bf #1}]}}
\newcommand{\ttnote}[1]{{\color{red}[\small {Tingting: \bf #1}]}}
\newcommand{\jmnote}[1]{{\color{blue}[\small {Jieming: \bf #1}]}}
\newcommand{\alenote}[1]{{\color{cyan}[\small {Alessandro: \bf #1}]}}

\else 
\newcommand{\tnote}[1]{}
\newcommand{\todo}[1]{}
\newcommand{\rc}[1]{}
\newcommand{\ttnote}[1]{}
\newcommand{\jmnote}[1]{}
\newcommand{\alenote}[1]{}

\fi

\renewcommand{\algorithmicrequire}{\textbf{Input:}}

\newcommand{\istar}{\ensuremath{{i^*}}\xspace}
\newcommand{\N}{\mathcal{N}}
\newcommand{\binjt}{\ensuremath{{\textsf{Bin}_j(t)}}}

\newcommand{\E}[0]{\mathop{\bbE}\xspace}

\DeclareMathOperator*{\poly}{poly}
 \newcommand{\occur}{\text{occurrency}}
\DeclareMathOperator*{\occurm}{occur}
  \newcommand{\eps}[0]{\ensuremath{\varepsilon}}
  \let\epsilon\eps
  
  \newcommand{\cA}{\ensuremath{{\mathcal A}}\xspace}
  \newcommand{\cB}{\ensuremath{{\mathcal B}}\xspace}
  
  \newcommand{\cD}{\ensuremath{{\mathcal D}}\xspace}

  \newcommand{\cN}{\ensuremath{{\mathcal N}}\xspace}

  \newcommand{\cR}{\ensuremath{{\mathcal R}}\xspace}
  \newcommand{\cS}{\ensuremath{{\mathcal S}}\xspace}
  
  \newcommand{\cU}{\ensuremath{{\mathcal U}}\xspace}

  \newcommand{\cX}{\ensuremath{{\mathcal X}}\xspace}

  \newcommand{\bbE}{\ensuremath{{\mathbb E}}\xspace}

  \newtheorem{definition}{Definition}
\newtheorem{theorem}{Theorem}
\newtheorem{lemma}{Lemma}
\newtheorem{corollary}[theorem]{Corollary}
\newtheorem{fact}{Fact}

\author{%
 Rachel Cummings\thanks{Supported in part by NSF grant CNS-1942772 and 2138834 (CAREER), a Google Cyber NYC Award, and the Center for Smart Streetscapes, an NSF Engineering Research Center, under grant agreement EEC-2133516.} \\
 Columbia University\\ \texttt{rac2239@columbia.edu}
 \and 
 Alessandro Epasto \\
 Google Research\\ \texttt{aepasto@google.com}
   \and
 Jieming Mao\\
   Google Research \\
   \texttt{maojm@google.com}
    \and
 Tamalika Mukherjee$^*$ \\
  Columbia University\\
   \texttt{tm3391@columbia.edu} 
 \and
 Tingting Ou$^*$\\
   Columbia University\\
  \texttt{to2372@columbia.edu} 
   \and
 Peilin Zhong\\
   Google Research\\
  \texttt{peilinz@google.com} 
 }

\title{Differentially Private Space-Efficient Algorithms for Counting Distinct Elements in the Turnstile Model}

\begin{document}

\maketitle
\begin{abstract}
The \emph{turnstile} continual release model of differential privacy captures scenarios where a privacy-preserving real-time analysis  is sought for a dataset evolving  through additions and deletions.  In typical applications of real-time data analysis, both the length of the stream $T$ and the size of the universe $|\mathcal{U}|$ from which data come can be extremely large.  This motivates the study of private algorithms in the turnstile setting using space sublinear in both $T$ and $|\mathcal{U}|$.  In this paper, we give the first sublinear space differentially private algorithms for the fundamental problem of counting distinct elements in the turnstile streaming model. Our  algorithm achieves, on arbitrary streams, $\tilde{O}_{\eta}(T^{1/3})$ space and additive error, and a $(1+\eta)$-relative approximation for all $\eta \in (0,1)$. Our result significantly improves upon the space requirements of the state-of-the-art algorithms for this problem, which is linear, approaching the known $\Omega(T^{1/4})$ additive error lower bound for arbitrary streams. Moreover, when a bound $W$ on the number of times an item appears in the stream is known, our algorithm provides $\Tilde{O}_{\eta}(\sqrt{W})$ additive error, using $\Tilde{O}_{\eta}(\sqrt{W})$ space. This additive error asymptotically matches that of prior work which required instead linear space.  
Our results address an open question posed by~\cite{JainKRSS23} about designing low-memory mechanisms for this problem. We complement these results with a space lower bound for this problem, which shows that any algorithm that uses similar techniques must use space $\tilde{\Omega}(T^{1/3})$ on arbitrary streams.
\end{abstract}

\section{Introduction}

Data streaming applications in which one needs to track specific statistics over a period of time occur in many real-world scenarios including digital advertising, network monitoring, and database systems. Given the sheer volume of data collected and processed in these online settings, significant work has been dedicated to the design of efficient algorithms to track and release useful statistics about the data stream, at every timestep. A key goal in this area is the design of space-efficient algorithms, i.e., algorithms that do not require storing the entire stream in memory. 

In many applications, the data being collected is not only massive, but also contains sensitive and personal user information. In this case, formal privacy protections are often required to ensure that the data released by the algorithm does not inadvertently leak protected information about individual users. Differential privacy (DP) \cite{dwork2006calibrating} has emerged as the de facto gold-standard in privacy-preserving data analysis to address such concerns. 

We focus on the \emph{turnstile} model of \emph{continual release} in differential privacy. In this model, we are given a stream $x=(x_1,\ldots, x_T)$, where each $x_t$ can be an insertion ($+u$) or deletion ($-u$) of some data item $u$ from a fixed universe $\cU$, or $\perp$ (no update), and the goal is to release a statistic of interest at every timestep $t \in [T]$, in a differentially private manner. An algorithm $\cA$ is $(\eps,\delta)$-\emph{differentially private} if for any neighboring input streams $x \sim x'$ and any output set $O$, $$\Pr[\cA(x) \in O] \leq e^\eps \Pr[\cA(x') \in O]+\delta.$$
In the continual release model, the output $\cA(x)$ refers to the entire output history of the algorithm $\cA$ over stream $x$ at every timestep. We consider event-level privacy, i.e., the input streams $x \sim x'$ are neighboring if they differ in at most one timestep.

In this paper, we study the fundamental problem of counting distinct elements in the turnstile model under continual release. 
Starting from the seminal works of~\cite{flajolet1985probabilistic, AlonMS99}, there is a rich literature of low space algorithms for the problem of counting distinct elements in the non-DP streaming world (e.g., see references within~\cite{muthukrishnan2005data}).

  While space-efficient DP algorithms for this problem have been recently studied in the (1) one-shot model (where the algorithm produces an output once at the end of the stream)~\cite{PaghMS2020, WangPS22, DickensTT22, HehirTC23, Smith0T20, StanojevicNY17, BravermanMWZ23, BlockiGMZ23}, and (2) insertion-only continual release model (where the algorithm outputs at every timestep, but no deletions are allowed)~\cite{EpastoMMMVZ23, DworkNPR10, ChanSS2011, BolotFMNT13,  Ghazi0NM23}, the existence of a private space-efficient algorithm in the more general turnstile model under continual release has been largely unexplored. Notably, even without any space constraints, recent work~\cite{JainKRSS23} has shown that designing DP algorithms for this problem in the general turnstile model is more challenging and can incur significantly more additive error, i.e., polynomial in $T$ in the turnstile model vs. polylogarithmic in $T$ in the insertion-only model. Understanding whether this fundamental problem can be solved in the turnstile model under continual release with reasonable accuracy and efficiency is an open problem. We address this gap by presenting the first sublinear space algorithm (sublinear in $T$ and $\vert \cU\vert$) for counting distinct elements.
  
For counting distinct elements in an input stream $x$ under continual release, one strategy used by prior work~\cite{EpastoMMMVZ23, JainKRSS23} is to approximately track the summation over the stream $s_x \in \{-1,0,1\}^T$ with the binary tree mechanism \cite{ChanSS2011,DworkNPR10}, where $s_x(t)$ is the difference in the number of distinct elements at timestep $t-1$ and $t$. Note that $ \sum_{i \leq t} s_x(i)$ precisely gives the number of distinct elements at timestep $t$. In the insertion-only setting, it is easy to observe that the sensitivity of the summation stream $s_x$ is bounded by a constant~\cite{EpastoMMMVZ23, BolotFMNT13}. In the turnstile model, \cite{JainKRSS23} observed that a change in the stream $x$ can cause $\Omega(T)$ changes to the corresponding summation stream $s_x$. In particular, the sensitivity of the summation stream $s_x$ in the turnstile model depends on the number of times an input stream item switches between being present to absent or vice versa --- this property is called the \emph{flippancy} of the item.~\cite{JainKRSS23} designed DP algorithms for counting distinct elements where the additive error scales with the flippancy of the stream, denoted $w_x$ (i.e., maximum flippancy over all items in the universe for a fixed stream). More specifically, the algorithm imposes a (provable) bound on the maximum flippancy of the stream.

However, flippancy is ill-suited to design low-space algorithms due to its inherently {\it stateful} nature (i.e., whether an element flips at time $t$ depends on all prior events of that element). For this reason, we introduce the related, but {\it stateless}, notion of \emph{occurrency} (see \cref{def:occur}) which measures the maximum number of times an item appears in the stream. Using this \occur\ measure, we design space-efficient DP algorithms for counting distinct elements with error and space that scales with the maximum \occur.

\begin{definition}[Occurrency]\label{def:occur}
Given a stream $x$ and an element $u \in \cU$, the \emph{\occur} of $u$ in $x$, denoted $\occurm(u,x)$, is the number of times an element $u$ appears (either as insertion or as  deletion) in the stream $x$. The \occur\ of the stream $x$, denoted $W_x$, is $\max_{u \in \cU} \occurm(u,x)$.    
\end{definition}

We remark that for any stream $x$ of size $T$, $W_x \leq T$; however, in many instances, the \occur\ is much smaller, and can be as low as $T$ / $|\cU|$.  For example, consider the case of an online learning platform that wants to estimate the number of activities being actively performed at the moment by users. A user can begin one activity (e.g., initiating an homework assignment for a course) and then terminate it when they are done (e.g., completing the assignment).\footnote{In this simple example, for clarity, we assume a user can only participate in the activity once. Our algorithms make no such assumption.} Each event in this example stream represents initiations/terminations of a given activity (where $+u$ is an initiation of an activity $u$, and $-u$ is a termination). In this case, \occur\ can be bounded by twice the maximum number of students in a class, which is clearly $\ll T$ (total number of events). 

In this paper, we first design a DP algorithm for counting distinct elements that, when provided with a stream promised to have \occur\ bounded by $W$, has an error and space that scales as a function of the promised bound $W$. When a bound on \occur\ is known, as in the example above,
this first algorithm can be used as is and will satisfy DP guarantees. However, in many cases an \occur\ bound may not be known. For this reason, we also design a more general algorithm that for an {\it arbitrary} stream of unknown \occur, {\it imposes} a bound on \occur\ for a subset of the stream. Combining these results, we prove our main result: a low-space algorithm with {\it unconditional} DP guarantees for arbitrary streams --- i.e., the algorithm is {\it truly private} in the sense of~\cite{JainKRSS23}. Finally, we give a matching space lower bound showing that any algorithm using a technique based on bounding \occur\ (or flippancy) cannot hope to achieve better space bounds.

\subsection{Our Contributions}

The main result of our paper is the design of the first sublinear space DP algorithms for the problem of counting distinct elements in the turnstile model under continual release. Before presenting our results, we first define some useful notation and terminology. We say $v$ is an $(\alpha,\beta)$-approximation of a function $f$ on input stream $x$ if $(1/\alpha) f(x) - \beta \leq   v \leq \alpha f(x) +\beta$.

We first present a DP algorithm for the problem of counting distinct elements, where the input stream has a promised upper bound on its \occur. In such cases, our DP algorithm achieves space and error guarantees that only have a sublinear dependency on \occur\ as stated below.\footnote{We present our informal results in terms of $(\eps,\delta)$-DP, however our formal theorems are stated in terms of $\rho$-zCDP. We can use \cref{thm:zcdp-to-dp} to convert $\rho$-zCDP to $(\eps,\delta)$-DP.} 

This algorithm utilizes a low-space dictionary data structure called $k$-set (denoted \textsf{KSET}), which was introduced by~\cite{Ganguly07} to estimate the number of distinct elements in the non-DP turnstile streaming setting. \textsf{KSET} supports the insertion and deletion of data items and returns (with high probability) the set of items $S$ present in the dictionary as long as $\vert S \vert \leq k$. If the number of elements exceeds $k$ or there is a collision where two elements are assigned the same cell in the \textsf{KSET}, then it ``fails'' by outputting NIL.

\begin{theorem}[Informal, \cref{thm:cd-cond-inf}]\label{thm:cond-inf}
For all $\eps, \eta>0$ and $\delta \in (0,1)$ and streams of length $T \in \mathbb{N}$, given a promised \occur\ bound of $W_x$,  
    there exists an $(\eps,\delta)$-DP algorithm in the  turnstile model under continual release that outputs a $(1 + \eta, \Tilde{O}_{\eps,\delta,\eta}(\sqrt{W_x} ))$-approximation to the number of distinct elements in the stream, using space $\Tilde{O}_{\eps,\delta,\eta}(\sqrt{W_x})$.
\end{theorem}

We remark that in this result, DP is guaranteed only when the input streams do not violate the promised \occur\ upper bound. Our next result allows us to remove this assumption and provide a DP algorithm for streams with unbounded \occur\ by creating a \emph{blocklist} of high \occur\ items to effectively ignore them in future timesteps. Before presenting our main result, we first define the problem of blocklisting high \occur\ items and show a near optimal space algorithm for this problem. 

Let  $\bf{blocklist}_{occ}(W)$ be the problem of outputting 0 at timestep $t$ when the current input element has \occur\ $<W$ (up to the  step before), and 1 otherwise. An algorithm for solving  $\bf{blocklist}_{occ}(W)$ reports a false negative if the element at timestep $t$ has \occur\ $\geq W$ but the algorithm outputs 0, and similarly reports a false positive if the element at timestep $t$ has \occur\ $<W$ but the algorithm outputs 1.

\begin{theorem}[Informal, Corollary~\ref{corol:blocklist-ub} and Theorem~\ref{thm:lb}]\label{th:informal-boundknown} 
    There exists an algorithm that, with high probability, has no false negatives and bounded false positives for   $\bf{blocklist}_{occ}(W)$ and uses space $\Tilde{O}(T/W)$. Moreover, for any $W>0$, any algorithm that solves  $\bf{blocklist}_{occ}(W)$ with the same false positive bound (and no false negatives) needs space $\Tilde{\Omega}(T/W)$. 
\end{theorem} 

Note that this space lower bound applies to any algorithm (even  exponential time ones) and implies that the space bounds achieved by our final algorithms are optimal up to log factors (among algorithms using the blocklisting approach). 

Finally, our main result gives a DP algorithm for counting distinct elements  with no assumptions on the \occur\ of the input stream. Our algorithm uses the blocklisting technique described above with the optimal choice of \occur\ $W_x=T^{2/3}$, combined with the \textsf{KSET} data strucure to achieve sublinear space while obtaining non-trivial additive error. 

\begin{theorem}[Main (Informal), \cref{thm:cd-uncond-inf}]\label{thm:main-inf}
For all $\eps, \eta>0$ and $\delta \in (0,1)$ and streams of length $T \in \mathbb{N}$, there exists an $(\eps,\delta)$-DP algorithm in the  turnstile model under continual release that outputs a $(1 + \eta,  \Tilde{O}_{\eps,\delta,\eta}(T^{1/3}))$-approximation to the number of distinct elements in the stream, using space $ \Tilde{O}_{\eps,\delta,\eta}(T^{1/3})$.
\end{theorem}

\paragraph{Comparison to prior work}
We observe that our results are the first to address the open question posed by~\cite{JainKRSS23} regarding the existence of accurate, private, and {\it low-memory} mechanism for counting distinct elements in turnstile streams.

When the stream has a promised flippancy bound $w_x$ rather than a promised \occur\ bound, the informal \cref{thm:cond-inf} (and the formal counterpart~\cref{thm:cd-cond-inf}) can be restated equally in terms of the flippancy of the stream. More precisely, our algorithm provides a $1 + \eta$ multiplicative approximation with $\Tilde{O}_{\eps,\delta,\eta}(\sqrt{w_x} ))$ additive error, using space $\Tilde{O}_{\eps,\delta,\eta}(\sqrt{w_x})$.\footnote{We omit the proof for simplicity as it is follows the same steps as our proof in terms of \occur.}

Notice that this additive error matches (neglecting lower order terms) the additive error $\Tilde{O}_\eps(\min(\sqrt{w_x},T^{1/3}))$ in the algorithm of~\cite{JainKRSS23}\footnote{Notice that sublinear space requires as well a multiplicative approximation factor in our algorithm}  which is achieved however with $\Omega(T)$ space. (This is because for the regime of $w_x \ge T^{2/3}$ one can use the unbounded \occur\ algorithm in our paper to obtain additive error $\Tilde{O}_{\eps,\delta,\eta}(T^{1/3})$). 

When flippancy is unbounded, instead, our additive error is always $\Tilde{O}_{\eps,\delta,\eta}(T^{1/3})$. This is close to the lower bound of $\tilde \Omega_\eps(T^{1/4})$ from~\cite{JainKRSS23} for private algorithms for the problem in streams of arbitrary large flippancy. Closing this gap is an interesting open problem.

\subsection{Related Work}
Our work on designing space-efficient turnstile streaming algorithms in the DP continual release setting is related to several topics in the areas of (non-private) space-efficient streaming algorithms, DP continual release, and DP streaming algorithms.

\textbf{Space-efficient streaming algorithms.} The streaming model of computation~\cite{flajolet1985probabilistic} is a well-known abstraction for efficient computation on large-scale data. In this model, data are received one input at a time, and the algorithm designer seeks to design algorithms that compute a solution on-the-fly while using limited space and time. In the area of non-private computation, a vast literature has been developed over the past decades~\cite{morris1978counting,flajolet1985probabilistic,flajolet2007hyperloglog,alon1996space} for addressing a variety of problems ranging from classical streaming computations (such as heavy hitters and frequency moments)~\cite{flajolet1985probabilistic,flajolet2007hyperloglog,durand2003loglog,cormode2005improved,misra1982finding}, to solving combinatorial optimization problems~\cite{mcgregor2018simple,HuangP19,charikar2003better}. Two key questions in the literature are whether it is possible to obtain space sublinear in the number of updates~\cite{flajolet2007hyperloglog} and whether it is possible to handle increasingly more dynamic updates (i.e., insertions-only, sliding window~\cite{datar2002maintaining}, and fully-dynamic streams~\cite{mcgregor2018simple}).  In the (non-private) streaming literature,  there is a well-developed theoretical understanding for interplay between the dynamicity of the stream, the accuracy achievable, and space bounds required, for a vast array of algorithmic problems~\cite{woodruff2004optimal}.

\textbf{DP continual release algorithms. } Differential privacy (DP)~\cite{dwork2006calibrating} has become the \emph{de facto} standard of private computation in algorithm design. In the context of streaming algorithms, the standard DP model is the \emph{continual release} model, first introduced by~\cite{DworkNPR10, ChanSS2011}. In this model, algorithms should preserve privacy of the input, even if a solution is observed at each timestep (as opposed to the \emph{one-shot} model where the adversary can obtain only \emph{one} solution at the end of the stream). Celebrated results in the continual release model include the well-known binary mechanism for releasing sum statistics in a binary data stream with $O(\eps^{-1}\log^2(T))$ additive error~\cite{DworkNPR10, ChanSS2011}. Significant work has expanded on this foundational result in various directions, including handling non-binary streams~\cite{TS13,FichtenbergerHO21}, sliding windows~\cite{BolotFMNT13}, and improving the space/utility tradeoffs~\cite{mcmahan2024efficient}. The latter work also provided algorithms for counting distinct elements with additive error of $O(\log^{1.5}(T))$ in insertion-only streams. Later work~\cite{Ghazi0NM23} also focused on counting distinct elements in the sliding window model, achieving polylogarithmic additive error in the window size. For insertion-only streams,~\cite{EpastoMMMVZ23} gave the first DP algorithms with space $O(\poly\log(T))$ and a $(1+\eta)$-multiplicative and  $O(\poly\log(T))$ additive error for counting distinct elements and frequency moment estimation in the insertion-only and sliding window model under continual release. 

\textbf{DP one-shot streaming algorithms.} A more restricted setting is the one-shot streaming model, where the analyst only seeks to output a solution at the end of the stream. In the one-shot setting,~\cite{DesfontainesLB19} showed  membership inference attacks for a large class of non-DP sketching algorithms for counting distinct elements, implying that they do not preserve any reasonable notion of privacy. In the private sphere and for a related problem, designing low-space DP algorithms for the general frequency moment estimation problem has been well-explored~\cite{Smith0T20, PaghMS2020, BlockiBDS12, DickensTT22, HehirTC23}. Specifically~\cite{WangPS22} showed that a well-known streaming algorithm called the $\mathbb{F}_p$ sketch preserves DP as is.~\cite{BlockiGMZ23} gave a black-box transformation for turning non-DP streaming algorithms into DP streaming algorithms while still preserving sublinear space and accuracy guarantees. These problems have also been explored in the \emph{pan-privacy} streaming model, where DP is preserved even if the internal memory of the algorithm is compromised~\cite{DworkNPR10, MirMNW11}.

\textbf{DP turnstile model.}
All previously mentioned work does not consider the fully dynamic streaming setting, where items can also be removed from the stream. The DP turnstile continual release model is the private equivalent of the fully-dynamic streaming model, and has received substantially less attention in the literature.
Counting distinct elements in the turnstile continual release model was only recently studied for the first time by~\cite{JainKRSS23}.  For a stream with flippancy bound $w$, they give an $(\eps,\delta)$-DP mechanism with additive error \sloppy $O(\eps^{-1}\sqrt{w} \poly\log(T)\sqrt{\log(1/\delta)})$ and space $\Omega(T)$.  They also show a lower bound of $\Omega(\min(w, {T}^{1/4}))$ on the additive error for any DP mechanism for this problem. In a concurrent work,~\cite{HenzingerSS23} also studied this problem under a restricted variant of the turnstile model where items are guaranteed to be present with cardinality at most $1$ at any time (i.e., multiple insertions of the same element are ignored). For this setting, they give an $(\eps,0)$-DP algorithm with additive error $\Tilde{O}(\sqrt{\eps^{-1}K\log(T)})$, where $K$ is the total number of insertions and deletions, and a nearly matching lower bound. Contrary to the non-private literature, theoretical understanding of space efficiency in private dynamic streaming algorithms is very limited. No prior work has designed differentially private sublinear-space algorithms for the foundational problem of counting distinct elements in the turnstile continual release model.

\section{Preliminaries}

We consider an input stream $x_1, \ldots, x_T$ of length $T$, coming from universe $\cU$, such that each $x_i \in \cU \cup \{\bot\}$. We assume that $\vert \cU \vert = \poly(T)$, as is standard in streaming literature~\cite{chakrabarti-streaming}. This assumption allows for simplified lower bound on space, since storing a single item from the universe requires $O(\log(|\cU|)) = \log(T))$ bits.

First, we recall the definition of differential privacy (DP) on streams. Neighboring streams $x$ and $x'$, denoted $x \sim x'$, differ in the stream elements at most one timestep.  

\begin{definition}[Differential privacy~\cite{dwork2006calibrating}] Given privacy parameters $\eps>0$ and $\delta \in [0,1)$, an algorithm $\cA$ is $(\eps,\delta)$-DP if for any neighboring streams $x \sim x'$ and any output set $O$, 
    $$\Pr[\cA(x) \in O] \leq e^\eps \Pr[\cA(x') \in O] +\delta.$$
\end{definition}
When $\delta=0$, this is known as pure DP, when $\delta>0$ and it is known as approximate DP. In the continual release setting, the output of $\cA(x)$ is the entire $T$-length output over the stream $x$ at every timestep.

Our privacy analysis will primarily use zero-concentrated differential privacy (zCDP), which is a slight variant of the standard DP definition. We present the definition of zCDP, its composition properties, and its relationship to $(\eps,\delta)$-DP in the one-shot setting, where the entire stream is processed and only one output is released at the end of the stream.

\begin{definition}[zero-concentrated differential privacy (zCDP) \cite{BunS16}]
    Given a privacy parameter $\rho>0$, a randomized algorithm $\cA$ satisfies $\rho$-zCDP if for all pairs of neighboring streams $x \sim x'$ and all $\alpha >1$, 
    $$\cD_\alpha (\cA(x) \| \cA(x') ) \leq \rho \alpha, $$
    where $\cD_\alpha ( P \| Q) = \frac{1}{\alpha-1}\log \left( E_{y \sim P} \left[ \frac{P(y)^{\alpha -1}}{Q(y)^{\alpha-1}} \right]\right)$ is the Renyi divergence of order $\alpha$ between probability distributions $P$ and $Q$. 
\end{definition}

There is also a relaxation of zCDP, known as \emph{approximate zCDP}, which is analogous to the relaxation between pure DP and approximate DP.

\begin{definition}[Approximate zCDP~\cite{BunS16}]\label{def:approx-zcdp}
    Given privacy parameters $\rho>0$ and $\delta \in (0,1)$, a randomized algorithm $\cA$ satisfies $\delta$-approximate $ \rho$-zCDP if for all pairs of neighboring streams $x \sim x'$, there exist events $E$ (which depends on $\cA(x)$) and $E'$ (which depends on $\cA(x')$) such that $\Pr [E] \geq 1-\delta$, $\Pr [E'] \geq 1-\delta$, and for all $\alpha \in (1,\infty)$,
    \begin{align*}
        D_\alpha (\cA(x)\vert_E \| \cA(x')\vert_{E'}) \leq \rho \alpha\ \vee D_\alpha (\cA(x')\vert_{E'} \| \cA(x)\vert_{E}) \leq \rho \alpha,
    \end{align*}
    where $\cA(x)\vert_E$ denotes the distribution of $\cA(x)$ conditioned on the event $E$.
\end{definition}

The privacy parameters of zCDP compose, similar to the composition guarantees of DP. Additionally, it is possible to translate between the guarantees of zCDP and DP.

\begin{theorem}[Composition \cite{BunS16}]\label{thm:zcdp-comp}
Let $\cA$ be a $\delta$-approximate $\rho$-zCDP algorithm and $\cA'$ be a $\delta'$-approximate $\rho'$-zCDP algorithm. Then the composition $\cA''(x)= (\cA(x), \cA'(x))$ satisfies $(\delta+\delta'-\delta\cdot \delta')$-approximate $(\rho+\rho')$-zCDP.
\end{theorem}

\begin{theorem}[Relationship to DP \cite{BunS16}]\label{thm:zcdp-to-dp}
For all $\rho,\delta >0$:
\begin{enumerate}
    \item If algorithm $\cA$ is $\rho$-zCDP then $\cA$ is $(\rho+2\sqrt{\rho \log(1/\delta)},\delta)$-DP. Conversely, if $\cA$ is $\eps$-DP then $\cA$ is $(\eps^2/2)$-zCDP. 
    \item If algorithm $\cA$ is $\delta$-approximate $\rho$-zCDP then $\cA$ is $(\eps, \delta+(1-\delta)\delta')$-DP for all $\eps \geq \rho$, where 
    $\delta' = \exp(-(\eps-\rho)^2/4\rho) \cdot \min \{ 1,
        \sqrt{\pi \cdot \rho},
        \frac{1}{1+{\eps-\rho)/2\rho }},
        \frac{2}{1+\frac{\eps-\rho}{2\rho} + \sqrt{(1+\frac{\eps-\rho}{2\rho})+\frac{4}{\pi\rho}}} \}.$
Conversely, if $\cA$ is $(\eps,\delta)$-DP then $\cA$ is $\delta$-approximate $(\eps^2/2)$-zCDP.  
\end{enumerate}
\end{theorem}

Finally, we present a simple mechanism that satisfies both DP and zCDP. The Gaussian Mechanism privately answers vector-valued queries by adding Gaussian noise to the true query answer. The noise is proportional to the \emph{$\ell_2$-sensitivity} of the function, which is the maximum change in the function's $\ell_2$-norm from changing a single element in the data.

\begin{definition}[Sensitivity \cite{dwork2006calibrating}]
    Let $f: \cX \to \mathbb{R}^k$ be a function. Its $\ell_2$-sensitivity is defined as 
    \begin{align*}
        \max_{x \sim x' \in \cX} \| f(x) - f(x') \|_2
    \end{align*}
\end{definition}

\begin{theorem}[Gaussian mechanism \cite{dwork2014algorithmic}] \label{thm:gaussian}
    Let $f: \mathcal{X}^n \to \mathbf{R}$ be a function with $\ell_2$-sensitivity at most $\Delta_2$. Let $\cA$ be an algorithm that on input $y$, releases a sample from $\mathcal{N}(f(y),\sigma^2)$. Then $\cA$ is $(\Delta^2_2/(2\sigma^2))$-zCDP. 
\end{theorem}

\section{Estimating the Number of Distinct Elements}\label{sec:distinct-elements}

In this section, we first present our DP algorithm (Algorithm~\ref{alg:large-universe-sampling}) in Section \ref{s.distinctalgo}, which outputs an approximation of the number of distinct elements in the stream in a continual release manner. We then present our main results for this algorithm in Section \ref{s.distinctresults}, namely its privacy, accuracy, and space guarantees, will details of the analysis deferred to Section \ref{sec:analysis-cd}. Finally, we present the main tools used in our analysis --- the \textsf{KSET} data structure in \cref{sec:kset} and the DP mechanism for counting in \cref{sec:bm-count-distinct}.

\subsection{Algorithm and Description}\label{s.distinctalgo}

Our algorithm \textsf{CountDistinct} (\cref{alg:large-universe-sampling}) for counting the number of distinct elements in a stream uses three main ingredients: (1) a \textsf{KSET} data structure which we use to store distinct elements from the stream $x$ with low space, (2) a binary-tree mechanism \textsf{BinaryMechanism-CD} for estimating the summation stream $s_x(t) \in \{-1,0,1\}^T$ which we obtain by comparing the cardinality of the distinct element set returned by the \textsf{KSET} data structure at timesteps $t-1$ and $t$, and (3) a blocklist $\cB$ which with high probability, stores all items whose \occur\ is too large.  

At a high-level, \textsf{CountDistinct} executes the following process on different subsamples. At each timestep $t\in [T]$, if the data element $x_t$ is non-empty, then the \textsf{COUNTING-KSET} subroutine updates the \textsf{KSET} data structure with $x_t$, and obtains the current set of distinct elements from the \textsf{KSET} (if the \textsf{KSET} does not fail; we discuss later how to deal with failures in the \textsf{KSET} as this case requires more care).  If the item corresponding to $x_t$ is present in the blocklist $\cB$, then the \textsf{COUNTING-KSET} subroutine does not update the \textsf{KSET} data structure with $x_t$. Next, the algorithm computes $s_x(t)$ to be the difference in counts between the current distinct elements set and the previously stored set, and feeds $s_x(t)$ into \textsf{BinaryMechanism-CD}, which produces a differentially privat count of the number of distinct elements  $\hat{s}_x(t)$. Finally $\hat{s}_x(t)$ is compared to a fixed threshold $\tau$. 
If $\hat{s}_x(t)$ is greater than $\tau$, then \textsf{COUNTING-KSET} returns TOO-HIGH, otherwise it returns $\hat{s}_x(t)$ to the main algorithm. The algorithm then performs the blocklisting step, which adds the current item to the blocklist $\cB$ with probability $p$ --- this ensures that the elements with high \occur\ are not likely to be considered by our algorithm in future iterations. 
This process is executed in $\log(T)$ parallel instances of \textsf{COUNTING-KSET} with different sampling rates using a hash function, in order to ensure that at least one sampling rate yields a good approximation to the number of distinct elements.

\begin{algorithm}[!ht]
\caption{\textsf{CountDistinct}}
\label{alg:large-universe-sampling}
\begin{algorithmic}[1]
\REQUIRE{ Stream $x_1, \ldots, x_T  \in \cU$, relative error $\eta \in (0,0.5)$, privacy parameter $\rho$, failure probability $\beta$,  boolean $ob$ that signals if we have an \occur\ bound on elements, \occur\ bound $W$
}
\STATE{ Let  $\ L \leftarrow \lceil \log (T) \rceil, \ \lambda \leftarrow 2\log (40 L/\beta)$ }
\IF{$ob$ is true}
\STATE{$\ \gamma = \sqrt{ \frac{  4  (W+1) (\log T + 1)^3  \log (10 (\log T + 1)/ \beta) }{\rho} }   $}
\ELSE
\STATE{$\ \gamma =  \sqrt{ \frac{  4  (T^{2/3} +1) (\log T + 1)^3  \log (10 (\log T + 1)/ \beta) }{\rho} }  + 3T^{1/3} \log (T^{1/3} \lceil \log T \rceil / \beta) $}
\ENDIF
\STATE{Let  $g:\cU \to [L]$ be a $\lambda$-wise independent hash function; for every $a \in \cU$, $i \in [L]$, $\Pr[g(a)=i] = 2^{-i}$, $\Pr[g(a)= \perp] = 2^{-L}$}\alglinelabel{line:hash}
\STATE{Initialize empty streams $\cS_1,\ldots, \cS_L$} \Comment{$\cS_i$ is the stream of noisy distinct counts and TOO-HIGHs}

\STATE{Initialize $\textsf{COUNTING-KSET}_1,\ldots,\textsf{COUNTING-KSET}_L$ with \occur\ bound $W$ if $ob=true$ or \occur\ bound $T^{2/3}$ if $ob=false$}\alglinelabel{line:kset} 
\STATE{Initialize blocklist $\cB = \emptyset$}

\FOR{ update $x_t$}

\FOR{$i\in [L]$}
\IF{$x_t \neq \perp$ and $g(x_t) = i$}
\STATE{$\cS_i[t] = $ \textsf{COUNTING-KSET}$_i$.\textsf{Update}($x_t, \cB$)\alglinelabel{line:replace-count} } \Comment{see \cref{alg:kset-count}}
\IF{$x_t \notin \cB$ and $ob$ is false}\label{line:blocklist-update}
\STATE{Add $x_t$ to $\mathcal{B}$ with probability $p = \frac{\log (T^{1/3} L/\beta)}{T^{2/3}}$} \alglinelabel{line:blocklist} 
\ENDIF

\ELSE
\STATE{$\cS_i[t] =$ \textsf{COUNTING-KSET}$_i$.\textsf{Update}($\bot,  \cB$)\alglinelabel{line:replace-count-empty}} \Comment{see \cref{alg:kset-count}}
\ENDIF
\ENDFOR
\STATE{\textbf{Output:} $\cS_i[t] \cdot 2^i$ for the largest $i \in [L]$ such that $\cS_i[t] $ is not TOO-HIGH and $\cS_i[t] \geq \max \{  \gamma/\eta, 32\lambda / \eta^2 \}$. (If such $i$ does not exist, \textbf{output} 0.) }\alglinelabel{line:outout-dist}
\ENDFOR
\end{algorithmic}
\end{algorithm}

\begin{algorithm}[!th]
\caption{\textsf{COUNTING-KSET}}
\label{alg:kset-count}
\begin{algorithmic}[1]
\REQUIRE{ Stream update $x_1, \ldots, x_T  \in \cU$, relative error $\eta \in (0,0.5)$, privacy parameter $\rho$, failure probability $\beta$, substream index $i$, \occur\ bound $W$ on elements, blocklist $\cB$
}
\STATE{ {Initialize} $\tau = 16 \max \{  \gamma/\eta, 32\lambda / \eta^2 \} +   2\sqrt{2} \frac{(\log T + 1)^{3/2}\sqrt{W \log(20T \lceil \log T \rceil /\beta)}}{\sqrt{\rho}}$}
\STATE{ Initialize $k= 16 \max \{  \gamma/\eta, 32\lambda / \eta^2 \} +   4\sqrt{2} \frac{(\log T + 1)^{3/2}\sqrt{W \log(20T \lceil \log T \rceil /\beta)}}{\sqrt{\rho}}$ }
\STATE{Initialize 
$\textsf{BinaryMechanism-CD}_i$ with parameters: privacy parameter $\rho/L$ and occurrency bound $W$ \alglinelabel{line:bm-init}}
\STATE{Initialize $\textsf{KSET}_i$ data structure with parameters: capacity $k$ and failure probability $ \beta / (2TL)$}\alglinelabel{line:kset-count-init}
\STATE{Initialize $F_{i,\textsf{last}} = 0$, $t_{last} = 0$ }  

\STATE{\textbf{Update}($x_t, \cB$):}
\FOR{ update $x_t$}
\IF{$x_t \neq \perp$ and $x_t \not\in \cB$}
\STATE{$\textsf{KSET}_i$.\textsf{Update}($x_t$) }  
\ENDIF
\STATE{Let $S_{i} \leftarrow \textsf{KSET}_i.\textsf{ReturnSet} $ } \Comment{Only keep the elements but not their counts }

\IF{$S_{i} \neq $ NIL }
\STATE{$t_{\textsf{diff}} = t - t_{\textsf{last}}$}
\STATE{$\textsf{diff}=|S_{i}| - F_{i,\textsf{last}}$\alglinelabel{line:sum-diff}}

\FOR{$j=1$ to $\vert \textsf{diff} \vert $}
\IF{$\textsf{diff} >0$}
\STATE{$\hat{s}_i \leftarrow $\textsf{BinaryMechanism-CD}$_i$.\textsf{Update}($1$)} 
\ELSIF{$\textsf{diff}<0$} 
\STATE{$\hat{s}_i \leftarrow $\textsf{BinaryMechanism-CD}$_i$.\textsf{Update}($-1$)}
\ENDIF
 \ENDFOR

 \FOR{$j=1$ to $t_{\textsf{diff}} - \vert \textsf{diff} \vert $}
 \STATE{$\hat{s}_i \leftarrow $\textsf{BinaryMechanism-CD}$_i$.\textsf{Update}($0$)}
 \ENDFOR

 \STATE{Save $F_{i,\textsf{last}} \leftarrow |S_{i}|$}
 \STATE{$t_{\textsf{last}}=t$} \alglinelabel{line:last}
 
 \ENDIF

\IF{$\hat{s}_i > \tau$ or $S_i = $NIL 
}\alglinelabel{line:threshold}
 \STATE{Return TOO-HIGH}
 \ELSE
\STATE{Return $\hat{s}_i$ } 
 \ENDIF
\ENDFOR
\end{algorithmic}
\end{algorithm}

More concretely, \cref{alg:large-universe-sampling} takes in a boolean flag $ob$ as input, indicating whether there is a promised bound on the \occur\ of the stream --- if $ob$ is true, the algorithm will operate under the assumption that the input stream has an \occur\ upper bounded by $W$, and thus does not employ the blocklisting technique. If $ob$ is false, the algorithm  imposes an internal bound of $W=T^{2/3}$ and executes the blocklisting procedure, in which it fixes a sampling probability $p$ (in Line~\ref{line:blocklist}), and every time an element occurs, the element is sampled with probability $p$ to be stored in a blocklist $\cB$. 

\cref{alg:large-universe-sampling}  uses a hash function to generate multiple parallel substreams $i \in [L]$ of the input stream (see Line \ref{line:hash}) where $L=\lceil \log(T) \rceil$, all subsampled with different sampling rates. 
The different sampling rates ensure that at least one sampling rate will yield a good approximation of the number of distinct elements in the original stream.

For each instance, $i\in L$, \cref{alg:large-universe-sampling} initializes the DP subroutine \textsf{COUNTING-KSET}$_i$ (\cref{alg:kset-count}), which uses two key subroutines:
\begin{enumerate}
    \item {\textsf{KSET} (Algorithm~\ref{alg:kset})}: The $k$-set structure is a dictionary data structure that supports insertion and deletion of data items and either returns, with high probability, the set of items $S$ that are present in the dictionary if $\vert S \vert \leq k$, or returns NIL (failure condition). Additional details are deferred to \cref{sec:kset}.
    \item {\textsf{BinaryMechanism-CD} (Algorithm~\ref{alg:bm})}:  This subroutine is used to privately count the sum of the difference in the count of distinct elements between consecutive timesteps. The mechanism is an extension of the Binary Mechanism~\cite{ChanSS2011, DworkNPR10}; however, the major differences are that it uses Gaussian noise (similar to~\cite{JainKRSS23}) and the input is $\{-1,0,1\}^T$ (as opposed to $\{0,1\}^T$ in the original). Although the \textsf{BinaryMechanism-CD} algorithm is similar to prior work, its privacy analysis needs to be handled carefully in our use-case, as it is closely tied to the failure behavior of the \textsf{KSET}. Additional details are deferred to \cref{sec:bm-count-distinct}.
\end{enumerate}

\textsf{COUNTING-KSET}$_i$ (\cref{alg:kset-count}) takes as input $x_t$ and updates the \textsf{KSET} data structure with $x_t$ (as long as $x_t$ is not in blocklist $\cB$ or equal to $\bot$). 
If the \textsf{KSET}$_i$ does not fail, then \textsf{COUNTING-KSET}$_i$ updates \textsf{BinaryMechanism-CD}$_i$ with the difference of the distinct sample size at time $t$ and the distinct sample size at the last timestep before $t_{last}< t$ (Line \ref{line:last}) that the \textsf{KSET} did not fail. Then \textsf{BinaryMechanism-CD}$_i$ outputs the noisy count of distinct elements, denoted $\hat{s}_i$. If the \textsf{KSET} does fail, i.e., $S_i=\textrm{NIL}$, then \textsf{COUNTING-KSET}$_i$ skips to Line \ref{line:threshold} which returns TOO-HIGH if the \textsf{KSET} fails, or $\hat{s}_i$ exceeds the threshold $k$. Note that this step is crucial for proving that \textsf{COUNTING-KSET}$_i$ is DP, which is presented in more detail in \cref{s.privacyresult}.  We note that \textsf{COUNTING-KSET}$_i$ does not take as input the flag $ob$, because if $ob=true$, then $\cB=\emptyset$ in \cref{alg:large-universe-sampling}, so taking $\cB$ as input is sufficient. 

Finally, \cref{alg:large-universe-sampling} maintains a stream of noisy distinct counts or TOO-HIGHs, denoted $\cS_i$, for each of the $\log T$ instances of \textsf{COUNTING-KSET}$_i$ instance. These streams are used to output $\cS_i[t] \cdot 2^i$ such that $\cS_i[t]$ is not TOO-HIGH in Line~\ref{line:outout-dist}, which is the private count of distinct elements in the $i$-th stream at time $t$, normalized by (the inverse of) that stream's sampling rate. This final count is output for all times $t \in [T]$.

\subsection{Main Results}\label{s.distinctresults}

In this subsection, we present our main results: \cref{thm:cd-cond-inf} and \cref{thm:cd-uncond-inf}, which summarize our main results on the privacy (\cref{thm:new_privacy}), accuracy (\cref{thm:acc_main}), and space (\cref{thm:space}) of \cref{alg:large-universe-sampling}. These results are all presented in greater detail in Section \ref{sec:analysis-cd}.

\cref{thm:cd-cond-inf} summarizes our main results for when there is a promised upper bound on the \occur\ of the input stream. 

\begin{corollary}\label{thm:cd-cond-inf}
   For all $\eta>0$ and $\beta \in (0,1)$, and a stream of length $T$, universe size $|\cU| = poly(T)$, and promised \occur\ $\leq W$, 
    there exists a $\beta$-approximate $\rho$-zCDP algorithm in the  turnstile model under continual release that, with probability $1-2\beta$, outputs a $(1 \pm \eta, \max \{ O(\gamma/\eta), O(\lambda / \eta^2)\})$-approximation to the number of distinct elements using space $O(\sqrt{W} \cdot \text{polylog}(T/\beta)) \cdot \text{poly}(\frac{1}{\rho \eta})$ where $\gamma =  O\left(\sqrt{ \frac{  W (\log T)^3  \log ( \log T/ \beta) }{\rho } }\right) $
    and $\lambda = O(\log (  \log (T) / \beta ))$.
\end{corollary}

The next result (\cref{thm:cd-uncond-inf}) summarizes our theoretical guarantees when the stream has unbounded \occur.

\begin{corollary}
\label{thm:cd-uncond-inf}\sloppy 
    For all $\eta>0$ and $\beta \in (0,1)$,  and a stream of length $T$ and universe size $|\cU| = poly(T)$, 
    there exists a $2\beta$-approximate $\rho$-zCDP algorithm in the  turnstile model under continual release that, with probability at least $1- 2\beta$ outputs a $(1 \pm \eta, \max \{ O(\gamma/\eta), O(\lambda / \eta^2)\})$-approximation to the number of distinct elements using space $O(T^{1/3} \cdot \text{polylog}(T/\beta)) \cdot \text{poly}(\frac{1}{\rho \eta})$ where  $ \gamma =  O\left(T^{1/3}\sqrt{ \frac{   (\log T)^3  \log ( \log T/ \beta) }{\rho } }\right)  $
    and $\lambda = O(\log (  \log (T)  / \beta ))$.
\end{corollary}

\subsection{\textsf{KSET} data structure}\label{sec:kset}

In this subsection, we present the \textsf{KSET} data structure, first introduced in \cite{Ganguly07}. In the non-DP turnstile setting, one key strategy for solving a variety of problems including estimating the number of distinct elements in sublinear space is to use space-efficient and dynamic {\it distinct sample} data structures. A distinct sample of a stream with sampling probability $p$ is a set of items such that each of the {\it  distinct} items in the stream has an equal and independent probability of $p$ of being included in the set. \cite{Ganguly07} introduced the \textsf{KSET} data structure for this problem, which can be used to (non-privately) give a $(1+\alpha)$-approximation of the number of distinct elements in a stream using roughly $O(\frac{1}{\alpha^2}(\log(T)+\log(\vert \cU \vert)\log(\vert \cU \vert))$ space. Our implementation of \textsf{KSET} is largely similar to the original version in \cite{Ganguly07}, with one small change that is necessary for our privacy analysis.

Our primary reason for using the \textsf{KSET} data structure in the DP setting is that it allows us to maintain a distinct sample. This in turn helps us add a smaller amount of noise (proportional to the bounded \occur\ of data items) to the final distinct elements estimate released through a binary mechanism. However, the privacy analysis for the output of the \textsf{KSET} is quite involved and discussed in \cref{s.privacyresult}. 
To the best of our knowledge, a distinct sample data structure has not previously been used in the DP setting.

The $k$-set structure (which we denote as \textsf{KSET} and is formally described in Algorithm~\ref{alg:kset}) 
is a dictionary data structure that supports insertion and deletion of data items and, with high probability, either returns the set of items $S$ present in the dictionary as long as $\vert S \vert \leq k$, or otherwise returns NIL. The structure is represented as a 2D array $H[R \times B]$ which consists of $R$ hash tables, each containing $B$ buckets, where $R=\lceil \log \frac{k}{\beta} \rceil$ and $B=2k$. For each $r \in [R], b\in [B]$, the bucket $H[r,b]$ contains a \textsf{TESTSINGLETON} data structure, which tests whether or not the bucket $H[r,b]$ contains a single universe element. Details on the implementation of the \textsf{TESTSINGLETON} data structure are deferred to \cref{app:kset}. 
 The $r$-th hash table in $H[R \times B]$ uses a pairwise independent hash function $h_r: \cU \to [B]$. Upon the arrival of an update $x_t$, the \textsf{KSET} structure contains two main operations:
 \begin{enumerate}
     \item The Update$(x_t)$ operation first increments (resp., decrements) the total number $m$ of data items in the structure (which is initialized to zero) if the update is an insertion (resp., deletion) of an item. Next, for every hash table $r \in [R]$, we update the corresponding \textsf{TESTSINGLETON} structure in the bucket $H[r,h_r(x_t)]$.   
     \item The ReturnSet$()$ operation, for every $r\in [R]$, iterates over the buckets $b \in [B]$, and checks whether the entry in the hash table $H[r,b]$ is a SINGLETON. If so, then it retrieves the data item along with its frequency and keeps track of the set of elements ($S$) as well as the total sum of frequencies of items ($m_s$) in $S$. 
If $m_s=m$ and $\vert S \vert \leq k$, then it returns $S$. Otherwise, the function returns NIL. We note that the latter check for $\vert S \vert \leq k$ is not included in the original version of \textsf{KSET} from~\cite{Ganguly07}, but it is crucial for our privacy analysis. We include this check to ensure that the \textsf{KSET} returns NIL with probability 1 in the event that there are more than $k$ elements. 
  \end{enumerate}
 
 Next, we state some simple properties regarding the accuracy of the \textsf{KSET} structure in Lemma~\ref{lem:k-set-acc} that are used in the analysis of Algorithm \ref{alg:large-universe-sampling}. 

\begin{algorithm}[!ht]
\caption{KSET data structure \cite{Ganguly07}}
\label{alg:kset}
\begin{algorithmic}[1]
\REQUIRE{ Capacity parameter $k$, failure probability $\beta$}
\STATE{Initialize 2D array $H[\log \frac{k}{\beta} \times 2k]$}
\STATE $R \leftarrow \lceil \log \frac{k}{\beta} \rceil, \ B \leftarrow 2k$
\STATE Let $h_r: \cU \rightarrow [B]$ be a pairwise independent function for $r = 1,\ldots, R$ 
\STATE Initialize $m = 0$
\STATE{\textbf{Update$(x_t)$:}} \Comment{Process update $x_t$}
\IF{$x_t$ is an insertion}
\STATE{$m \leftarrow m + 1$}
\ELSIF{$x_t$ is a deletion}
\STATE{$m \leftarrow m - 1$}
\ENDIF
\FOR{$r \in [R]$}
\STATE{$H[r, h_r(x_t)]$.\textsf{TSUPDATE}($x_t$)  } \Comment{see Algorithm \ref{alg:ts}}
\ENDFOR
\STATE{\textbf{{ReturnSet():} }}
\STATE{Initialize set $S = \{ \}$, $m_s = 0$}
\FOR{$r\in [R]$}
\FOR{$b \in [B]$}
\IF{$H[r,b]$.\textsf{TSCARD()[0]} == \textsc{SINGLETON}}  
\STATE{$(x, c) \leftarrow$ $H[r,b]$.\textsf{TSCARD()[1]}, $H[r,b]$.\textsf{TSCARD()[2]} }  \Comment{see Algorithm \ref{alg:ts}}
\STATE{Insert $(x,f_x)$ to $S$}
\STATE{$m_s \leftarrow m_s + f_x$}
\ENDIF
\ENDFOR
\ENDFOR
\IF{$m_s = m$ and $|S| \leq k$}\label{line:kset-distinct}
\STATE{Return $S$}
\ELSE 
\STATE{Return NIL}
\ENDIF
\end{algorithmic}
\end{algorithm}

\begin{lemma}[KSET properties] \label{lem:k-set-acc}
Consider a \textsf{KSET} data structure with capacity $k$ and failure probability $\beta$. Then, the following holds:
\begin{enumerate}
\item If strictly more than $k$ distict elements are present in the \textsf{KSET}, then with probability 1, $\textsf{KSET}.\textsf{ReturnSet} =$ NIL. \label{it:kset-full}
\item If less than or equal to $k$ distict 
 elements are present in the \textsf{KSET} then with probability $\geq 1-\beta$, $\textsf{KSET}.\textsf{ReturnSet} = S$ where $S$ is the entire set of items present in the \textsf{KSET}. \label{it:kset-s}
\item The space complexity of \textsf{KSET} is $O(k (\log T + \log |\mathcal{U}|) \log \frac{k}{\beta})$.\label{it:kset-space}
\end{enumerate}
\end{lemma}
\begin{proof}
For \cref{it:kset-full}, from~\cite{Ganguly07}, we have that if $m_s=m$ then the set of retrieved items $S$ is exactly the set of distinct elements with probability $1$.  Observe that if there are more than $k$ elements, then either $m_s \neq m$, or the number of distinct elements returned is $\vert S \vert>k$. In both cases, the condition in \cref{line:kset-distinct} fails and the output is NIL. 

    The proofs for \cref{it:kset-s} and \cref{it:kset-space} are identical to the one given in~\cite{Ganguly07}.
\end{proof}

\subsection{Binary Mechanism for the Count Distinct Problem}\label{sec:bm-count-distinct}

In this subsection, we present a subroutine for \textsf{COUNTING-KSET} (\cref{alg:kset-count}) that is a modified Binary Mechanism called \textsf{BinaryMechanism-CD} (see \cref{alg:bm}). This subroutine is used to compute the summation stream $s_x \in \{-1,0,1\}^T$ representing the difference in the number of distinct elements at timesteps $t-1$ and $t$ which is computed from the output of the \textsf{KSET} in Line~\ref{line:sum-diff} of \cref{alg:kset-count}. The algorithm \textsf{BinaryMechanism-CD} injects Gaussian noise (similar to~\cite{JainKRSS23}) --- as opposed to Laplace noise used in the original versions of the Binary Mechanism~\cite{ChanSS2011, DworkNPR10} --- proportional to the sensitivity of the summation stream. 

\begin{algorithm}[!ht]
\caption{\textsf{BinaryMechanism-CD}}
\label{alg:bm}
\begin{algorithmic}[1]
\REQUIRE{Count distinct summation stream $y_1, y_2, \ldots, y_T \in \{ -1, 0, 1\}^T$, privacy parameter $\rho > 0$, \occur\ $W > 0$}
\STATE Initialize each $\alpha_i=0$ and $\hat{\alpha_i}=0$
\STATE{Let $\rho' = \frac{\rho}{2(W+1) (\log T + 1)}$}
\STATE{\textbf{Update}($y_t$):}
\FOR{every update $y_t$}
\STATE{Express $t$ in binary from: $t = \sum_j \binjt 2^j$}
\STATE{Let $i = \min \{ j : \binjt = 1\}$} be the least significant binary digit, and set $\alpha_i = \sum_{j=0} ^{i-1}
\alpha_j +  y_t $
\FOR{$j = 0, 1, \ldots, i-1$}
\STATE Set $\alpha_j = 0$ and $\hat{\alpha}_j = 0$
\ENDFOR
\STATE Set $\hat{\alpha}_i = \alpha_i + \N(0, 1/\rho')$
\STATE \textbf{Return}
$\mathcal{B}(t) = \sum _{j : \binjt = 1 } \hat{\alpha} _j $
\ENDFOR
\end{algorithmic}
\end{algorithm}

We first note that if the input to \textsf{BinaryMechanism-CD} represented the \emph{exact} difference in the number of distinct elements over consecutive timesteps, then the sensitivity of the summation stream in terms of \occur\ can be calculated in a straightforward manner using arguments similar to~\cite{JainKRSS23}. 
However, in for our use of \textsf{BinaryMechanism-CD} in the algorithm \textsf{COUNTING-KSET}, \textsf{BinaryMechanism-CD} receives as input the  difference in the number of distinct elements from the output of the \textsf{KSET}. Importantly, the failure behavior of the \textsf{KSET} needs to be accounted for when arguing about the sensitivity of the resulting summation stream computed from the \textsf{KSET} output (when it does not fail). In order to do this, we use a coupling argument to show that the output stream of the algorithm \textsf{COUNTING-KSET} is close to the output stream of an algorithm that exactly computes the number of distinct elements and feeds the difference over consecutive timesteps to \textsf{BinaryMechanism-CD} as input. The privacy analysis for \textsf{BinaryMechanism-CD} in the latter algorithm is similar to~\cite{JainKRSS23} and is given in \cref{lem:bm-ec-dp}. The privacy guarantee of our application of \textsf{BinaryMechanism-CD} inside \textsf{COUNTING-KSET} is implicitly derived in the privacy analysis of \textsf{COUNTING-KSET} via the coupling argument of \cref{lem:new_privacy} and the claim that \textsf{COUNTING-KSET} is DP in \cref{cor:kset-count-dp}.    

The accuracy guarantee of \textsf{BinaryMechanism-CD} follows from \cref{lem:bm_acc} in which we consider the overall accuracy of $L$ instances of \textsf{BinaryMechanism-CD} as instantiated in Line~\ref{line:bm-init} of \textsf{COUNTING-KSET}. Finally, the space complexity of \textsf{BinaryMechanism-CD} is $O(\log(T))$ and this follows from~\cite{ChanSS2011}.

\section{Analysis of \textsf{CountDistinct}}\label{sec:analysis-cd}

This section presents details of the privacy (\cref{thm:new_privacy}), accuracy (\cref{thm:acc_main}) and space (\cref{thm:space}) guarantees of \cref{alg:large-universe-sampling}. Omitted proofs can be found in \cref{app.proofs}.

\subsection{Privacy}\label{s.privacyresult}

First we present the privacy analysis, showing that \textsf{CountDistinct} (\cref{alg:large-universe-sampling}) is differentially private. One key challenge in the privacy analysis is to ensure that the output of the \textsf{KSET} data structure does not leak privacy, even in the event of its failure (i.e., if more than $k$ distinct elements are stored). Note that if the \textsf{KSET} \emph{never} failed, then one could simply sum up the difference in the output sizes of the \textsf{KSET} over consecutive timesteps, i.e., the stream $s_x \in \{-1,0,1\}^T$ using \textsf{BinaryMechanism-CD}. Assuming no failures, the sensitivity of $s_x$ can be bounded in terms of the maximum \occur\ $W$, since in this case the \textsf{KSET} returns exact counts. However, the failure of \textsf{KSET} cannot be avoided or absorbed into the $\delta$ parameter because the \textsf{KSET} will fail with probability 1 when its capacity exceeds $k$ (see \cref{lem:k-set-acc}).

We address this challenge by modifying the original \textsf{KSET} algorithm to have well-behaved failures. More precisely, we introduce a thresholding step where our algorithm returns TOO-HIGH if the \textsf{KSET} fails or approaches a regime where failing is a likely event. The latter can be estimated privately by verifying whether the binary tree output is too large, $\hat{s}_x>\tau$, and returning TOO-HIGH if so (see \textsf{COUNTING-KSET}, \cref{alg:kset-count}).

By doing so, we can make a coupling argument between our algorithm and a much simpler algorithm (\textsf{COUNTING-DICT}, \cref{alg:exact-count}) that simply stores the exact counts of elements and computes $\hat{s}_x$ via \textsf{BinaryMechanism-CD} and has the same thresholding step: if $\hat{s}_x>\tau$, return TOO-HIGH. Note that \textsf{COUNTING-DICT} is not space-efficient, and we only introduce it for analysis purposes. In \textsf{COUNTING-DICT}, the sensitivity of  $s_x$ can be bounded in terms of the maximum \occur\ $W$ for all timesteps. We then use the coupling argument to bound the sensitivity of $s_x$ in \textsf{COUNTING-KSET} for all timesteps (barring specific bad events whose failure probability is negligible and absorbed into the DP failure probability).

\begin{algorithm}[!ht]
\caption{\textsf{COUNTING-DICT}}
\label{alg:exact-count}
\begin{algorithmic}[1]
\REQUIRE{Stream $x_1, \ldots, x_T  \in \cU$, relative error $\eta \in (0,0.5)$, privacy parameter $\rho$, failure probability $\beta$, \occur\ bound $W$ on elements, substream index $i$, list of blocklisted elements $\cB$ 
}
\STATE{Let $\hat{s}_i=0$ and $\tau = 16 \max \{  \gamma/\eta, 32\lambda / \eta^2 \} +   2\sqrt{2} \frac{(\log T + 1)^{3/2}\sqrt{W \log(20T \lceil \log T \rceil /\beta)}}{\sqrt{\rho}}$  } 
\STATE{Initialize 
$\textsf{BinaryMechanism-CD}_i$ with parameters: privacy parameter $\rho/L$ and occurrency bound $W$ }
\STATE{Initialize $\textsf{DICT}_i$ dictionary data structure of size $\vert \cU \vert \times T$}
\STATE{Initialize $F_{i, last} = 0$ }
\STATE{\textbf{Update}($x_t, \cB$):}
\FOR{ update $x_t$}
\IF{$x_t \neq \perp$ and $x_t \not\in \cB$ }
\IF{$x_t$ is an insertion}\label{line:if-ec}
\STATE{$\textsf{DICT}_i[x_t ][t]= \textsf{DICT}_i[ x_t ][t-1] + 1$ }  
\ELSE
\STATE{$\textsf{DICT}_i[x_t ][t]= \textsf{DICT}_i[ x_t ][t-1] - 1$ }  
\ENDIF \label{line:ifend-ec}
\STATE{Let $s_{i} \leftarrow \sum_{u \in \cU} \mathbf{1}_{\textsf{DICT}_i[u][t]>0}  $}\label{line:s-ec}
\STATE{$\hat{s}_i \leftarrow $\textsf{BinaryMechanism-CD}$_i$.\textsf{Update}($s_{i} - F_{i,last}$)}
\STATE{Save $F_{i,last} \leftarrow { s_{i}}$}
\ELSE
\STATE{$\hat{s}_i \leftarrow $\textsf{BinaryMechanism-CD}$_i$}.\textsf{Update}(0) \ENDIF
\IF{$\hat{s}_i > \tau$}
 \STATE{Return TOO-HIGH }
 \ELSE
\STATE{Return $\hat{s}_i$} 
 \ENDIF
\ENDFOR
\end{algorithmic}
\end{algorithm}

We present the main theorem of the privacy guarantee below. The proof requires showing that the output stream published by \textsf{COUNTING-KSET}$_i$ is DP (see \cref{cor:kset-count-dp}), which is argued by showing that the outputs of \textsf{COUNTING-KSET}$_i$ and \textsf{COUNTING-DICT}$_i$ are identical except with probability at most $\beta$ (see \cref{lem:new_privacy}). Since all the operations after calling the subroutine \textsf{COUNTING-KSET}$_i$ in Lines \ref{line:replace-count} and \ref{line:replace-count-empty} of \cref{alg:large-universe-sampling} is post-processing, we will have shown that \textsf{CountDistinct} is approximate zCDP (see~\cref{def:approx-zcdp}) in \cref{thm:new_privacy}. 

\begin{restatable}{theorem}{privacythm}\label{thm:new_privacy}
    \textsf{CountDistinct} (\cref{alg:large-universe-sampling}) is 
    \begin{enumerate}
        \item \label{it:dp-cond} $\beta$-approximate $\rho$-zCDP if $ob=true$,
        \item \label{it:dp-uncond} $2\beta$-approximate $\rho$-zCDP if $ob=false$. 
    \end{enumerate}
\end{restatable}

To prove \cref{thm:new_privacy}, we first establish that the outputs of \textsf{COUNTING-KSET}$_i$ and \textsf{COUNTING-DICT}$_i$ are identical with high probability (\cref{lem:new_privacy}, proven in Appendix \ref{app.newprivacyproof}). 
Define \textsf{CountDistinct'} as a variant of \textsf{CountDistinct} that replaces calls to \textsf{COUNTING-KSET}$_i$ (\cref{alg:kset-count}) in Line~\ref{line:kset}, Line~\ref{line:replace-count}, Line~\ref{line:replace-count-empty} with calls to \textsf{COUNTING-DICT}$_i$ (\cref{alg:exact-count}) for $i\in [L]$.

\begin{restatable}{lemma}{ksetexactmatch}\label{lem:new_privacy}
Fix the randomness used across runs of \textsf{CountDistinct} and \textsf{CountDistinct}'. 
Fix $i \in [L]$, and let $K$ and $E$ denote the output distributions of \textsf{COUNTING-KSET}$_i$ and \textsf{COUNTING-DICT}$_i$ respectively. 
 If $k \geq \tau + O\left( \frac{\text{polylog}  (T/\beta) \sqrt{W} }{\sqrt{\rho}} \right)$, then the total variation distance of the two distributions, $d_{TV}(K,E) \leq \beta/L$.
\end{restatable}

To prove that \textsf{COUNTING-KSET}$_i$ is differentially private (Corollary~\ref{cor:kset-count-dp}), we first show that \textsf{COUNTING-DICT}$_i$ satisfies differential privacy (\cref{lem:exact_count_dp}).

\begin{restatable}{lemma}{exactcountdp}\label{lem:exact_count_dp}
\textsf{COUNTING-DICT}$_i$ (\cref{alg:exact-count}) is 
\begin{enumerate}
    \item  $\rho/L$-zCDP, if $ob=true$.
    \item $\beta/L$-approximate $\rho/L$-zCDP, if $ob=false$.
\end{enumerate}
\end{restatable}

The full proof of Lemma \ref{lem:exact_count_dp} is in Appendix \ref{app.exactcountproof}; we provide a proof sketch here to highlight the key ideas. 
Recall that \textsf{COUNTING-DICT}$_i$ uses a dictionary data structure \textsf{DICT} to store the counts of the elements seen in the stream, and the difference in the number of distinct elements is exactly computed from \textsf{DICT} and fed as input to \textsf{BM-Count-Distinct}$_i$, which outputs a noisy distinct element count $\hat{s}_i$. Thus, in order to prove the privacy guarantee of \textsf{COUNTING-DICT}$_i$, we need to show that \textsf{BM-Count-Distinct}$_i$ is DP. Then the output of \textsf{COUNTING-DICT}$_i$ will simply be post-processing on the DP output of \textsf{BM-Count-Distinct}$_i$ (either $\hat{s}_i$ or TOO-HIGH), so it will also be DP. 
Hence, we prove that \textsf{BM-Count-Distinct}$_i$ when used in \textsf{COUNTING-DICT}$_i$ is $\beta/L$-approximate $\rho/L$-zCDP (in \cref{lem:bm-ec-dp}), which implies Lemma \ref{lem:exact_count_dp}.

We next sketch the proof of \cref{lem:bm-ec-dp}, that \textsf{COUNTING-DICT}$_i$ is $\rho/L$-zCDP.
On a high-level, we need to argue that if $x$ and $x'$ are neighboring streams, then the resulting streams (after hashing and blocklisting, see \cref{def:bl-stream}) that are fed as input to \textsf{BM-Count-Distinct}$_i$ can differ in at most $W+1$ positions with probability $1-\beta/L$, for \occur\ bound $W=T^{2/3}$ (see \cref{lem:bl-id}). Opening up the analysis of the binary tree mechanism in \textsf{BM-Count-Distinct}$_i$, we show that the sensitivity of the nodes over all levels of the binary tree is at most $2 \sqrt{(W+1)(\log(T)+1)}$. Thus adding Gaussian noise proportional to this quantity to each node of the binary tree preserves $\rho/L$-zCDP with probability $\beta/L$.

We emphasize that the privacy argument for the \textsf{BM-Count-Distinct}$_i$ instance in \textsf{COUNTING-DICT}$_i$ cannot be directly applied to the \textsf{BM-Count-Distinct}$_i$ instance in \textsf{COUNTING-KSET}$_i$. This is because in \textsf{COUNTING-KSET}$_i$, the output of the \textsf{KSET} is used to compute the input stream to \textsf{BM-Count-Distinct}$_i$, and failures of the \textsf{KSET} can lead to a large difference in the outputs of \textsf{BM-Count-Distinct}$_i$ on neighboring streams.
Thus the failure behavior of the \textsf{KSET} must be handled carefully in the privacy analysis. \cref{cor:kset-count-dp} gives the resulting privacy guarantee for \textsf{COUNTING-KSET}$_i$.

\begin{corollary}\label{cor:kset-count-dp}
\textsf{COUNTING-KSET}$_i$ (\cref{alg:kset-count}) is (1) $(\beta/L)$-approximate $(\rho/L)$-zCDP if $ob=true$, and (2) $(2\beta/L)$-approximate $(\rho/L)$-zCDP if $ob=false$.
\end{corollary}

\begin{proof}
    From \cref{lem:exact_count_dp} we know that \textsf{COUNTING-DICT}$_i$ is $\beta/L$-approximate $\rho/L$-zCDP. In particular, for $ob=true$, \textsf{COUNTING-DICT}$_i$ is $\rho/L$-zCDP. From \cref{lem:new_privacy}, we know that the output distribution of \textsf{COUNTING-DICT}$_i$ and \textsf{COUNTING-KSET}$_i$ is identical except with probability $\beta/L$. The claims for when $ob=true$ vs $ob=false$ follows.
\end{proof}

We are finally ready to prove \cref{thm:new_privacy}.

\begin{proof}[Proof of \cref{thm:new_privacy}]

We first argue about the case when $ob=false$ as this is the more general case. 

   The randomness of \textsf{CountDistinct} can be viewed as a joint probability distribution $\cR_{CD}= \cR_{g} \times \cR_{BL} \times \cR_{KC_1}\times \ldots \times \cR_{KC_L}$ where $\cR_g$ denotes the randomness from picking hash function $g$ (in Line~\ref{line:hash} of \cref{alg:large-universe-sampling}), $\cR_{BL}$ denotes the randomness associated with sampling an element to add to blocklist $\cB$ (in Line~\ref{line:blocklist}), and $\cR_{KC_i}$ denotes the randomness from the subroutine \textsf{COUNTING-KSET}$_i$ for $i\in [L]$. Similarly, the randomness of \textsf{CountDistinct}' can be viewed as a joint probability distribution $\cR_{CD'}= \cR_{g} \times \cR_{BL} \times \cR_{EC_1}\times \ldots \times \cR_{EC_L}$ where $\cR_g$ denotes the randomness from picking a hash function $g$, $\cR_{BL}$ denotes the randomness associated with sampling an element to add to blocklist $\cB$ and $\cR_{EC_i}$ denotes the randomness from the subroutine \textsf{COUNTING-DICT}$_i$ for $i\in [L]$. 

    We first define an identity coupling over the randomness $\cR_g$ of picking the hash function and the randomness $\cR_{BL}$ of blocklisting between \textsf{CountDistinct} and \textsf{CountDistinct'}.  
    In other words, we fix the same hash function $g$ and the same sampling rate to blocklist an item for \textsf{CountDistinct} and \textsf{CountDistinct'}. 
    Applying \cref{lem:new_privacy}, we have that the outputs of \textsf{COUNTING-KSET}$_i$ and \textsf{COUNTING-DICT}$_i$ are identical except with probability $\beta/L$. In particular, from \cref{cor:kset-count-dp}, we have that \textsf{COUNTING-KSET}$_i$ is $2\beta/L$-approximate $\rho/L$-zCDP. 
    Since we have a total of $L$ substreams, using a union bound argument, the entire \textsf{CountDistinct} algorithm is $2\beta $-approximate $\rho$-zCDP by basic composition of zCDP (\cref{thm:zcdp-comp}).\footnote{Note that \cref{thm:zcdp-comp} gives an even tighter guarantee, but we we use this slightly weaker composition for a cleaner presentation.}

    In the case when $ob=true$, note that we only have to consider the identity coupling of over hashing items using hash function $g$ in \textsf{CountDistinct} and \textsf{CountDistinct'}. 
    The claim that the outputs are identical except with probability $\beta/L$ follows from \cref{lem:new_privacy}. But now, from \cref{cor:kset-count-dp}, we have that \textsf{COUNTING-KSET}$_i$ is $\beta/L$-approximate $\rho/L$-zCDP, and the rest of the argument follows from a union bound and composition.   
 \end{proof}

\subsection{Accuracy}\label{s.accresult}

We present our main accuracy theorem in Theorem~\ref{thm:acc_main}. Note that our accuracy theorem is also parameterized by the value of the boolean flag $ob$ and gives different error/space trade-offs according to whether $ob$ is true or false. In particular if $ob$ is true, meaning our algorithm is promised that the \occur\ of input streams is bounded by $W$, then we get additive error only that has $\sqrt{W}$ dependency. If $ob$ is false, then our algorithm makes no assumption on the \occur\ of the input stream and therefore incurs a higher additive error.

\begin{theorem}[Accuracy of Algorithm~\ref{alg:large-universe-sampling}]\label{thm:acc_main}
Let $F(t)$ be the correct number of distinct elements of the stream at time $t$ and let $\lambda = 2\log (40 \lceil \log (T) \rceil / \beta )$. 
\begin{enumerate}
    \item \label{it:acc-cond} When $ob$ is true, let $\gamma =  \sqrt{ \frac{  4  (W+1) (\log T + 1)^3  \log (10 (\log T + 1)/ \beta) }{\rho} } $, and
    \item \label{it:acc-uncond} when $ob$ is false, let $\gamma = \sqrt{ \frac{  4  (T^{2/3}+1) (\log T + 1)^3  \log (10 (\log T + 1)/ \beta) }{\rho} }  + 3T^{1/3} \log ( T^{1/3} \lceil \log T \rceil / \beta)$.
\end{enumerate}
For a fixed timestep $t \in [T]$, with probability at least $1- 2\beta$, the output of Algorithm~\ref{alg:large-universe-sampling} at time $t$ is a $(1 \pm 4\eta, 32 \max \{ \gamma/\eta, 32\lambda / \eta^2\})$-approximation of $F(t)$ for any $\eta \in (0, 0.5)$.
\end{theorem}

To prove Theorem \ref{thm:acc_main}, we use three helper lemmas, all proved in Appendix \ref{app.helperlemmas}. Lemma \ref{lem:substream_concentration} bounds the number of elements in the substream after hashing. Lemma \ref{lem:bm_acc} proves the accuracy of \textsf{BinaryMechanism-CD} algorithm. Lemma \ref{lem:blocklist_size} bounds the size of the blocklist when $ob$ is false. With the help of these lemmas, we can show the accuracy of \textsf{COUNTING-DICT}, as an intermediate step in the analysis. The proof of Theorem \ref{thm:acc_exact} is deferred to Appendix \ref{app:exactacc}.

\begin{restatable}{theorem}{accuracy}
\label{thm:acc_exact}
Let $F(t)$ be the correct number of distinct elements of the stream at time $t$ and let $\lambda = 2\log (40 \lceil \log (T) \rceil / \beta )$. When $ob$ is true, let $\gamma =  \sqrt{ \frac{  4  (W+1) (\log T + 1)^3  \log (10 (\log T + 1)/ \beta) }{\rho} } $ and when $ob$ is false, let $\gamma = \sqrt{ \frac{  4  (T^{2/3}+1) (\log T + 1)^3  \log (10 (\log T + 1)/ \beta) }{\rho} }  + 3T^{1/3} \log (T^{1/3} \lceil \log T \rceil / \beta) $.
For a fixed timestep $t \in [T]$, with probability at least $1- \beta$, the output of Algorithm~\ref{alg:large-universe-sampling} at time $t$ with \textsf{COUNTING-DICT}$_i$ as the subroutine is a $(1 \pm 4\eta, 32 \max \{ \gamma/\eta, 32\lambda / \eta^2\})$-approximation of $F(t)$ for any $\eta \in (0, 0.5)$.
\end{restatable}

With this result, we can finally prove Theorem \ref{thm:acc_main}, by showing that substituting \textsf{COUNTING-KSET} in place of \textsf{COUNTING-DICT} still allows high accuracy of Algorithm~\ref{alg:large-universe-sampling}.

\begin{proof}[Proof of Theorem \ref{thm:acc_main}]
    We apply a union bound argument that combines Lemma~\ref{lem:new_privacy} and Theorem~\ref{thm:acc_exact}. By Lemma~\ref{lem:new_privacy}, we can link the accuracy of \textsf{COUNTING-DICT}$_i$ to that of \textsf{COUNTING-KSET}$_i$ and with probability at least $1-\beta$, the output distributions of all the $L$ instances of \textsf{COUNTING-DICT}$_i$ and \textsf{COUNTING-KSET}$_i$ used in Algorithm~\ref{alg:large-universe-sampling} are the same. Furthermore, we have argued that we have the desired accuracy with probability at least $1-\beta$ for Algorithm~\ref{alg:large-universe-sampling} with \textsf{COUNTING-DICT}$_i$ in Theorem~\ref{thm:acc_exact}. Thus, by a union bound over the two events that (1) the output of \textsf{COUNTING-DICT}$_i$ matches the output of \textsf{COUNTING-KSET}$_i$ for $i\in [L]$ and (2) \textsf{COUNTING-DICT}$_i$ is accurate, we have proven our claim.   
\end{proof}

\subsection{Space complexity}\label{s.spaceresult}

Finally, we prove the space guarantees of our algorithm below. As with the privacy and accuracy result, \cref{thm:space} is parameterized by the value of the boolean flag $ob$. When $ob$ is true and the input stream is promised to have \occur\ bounded by $W$, the space is only polynomial in $W$. When $ob$ is false, the algorithm allows general input streams, and internally enforces a bound $T^{2/3}$ on the stream's \occur\ using blocklisting, which requires more space.

\begin{theorem}\label{thm:space}
With probability at least $1-\beta$, assuming the universe size $|\cU| = \text{poly}(T)$:
\begin{enumerate}
       \item\label{it:space-cond} If $ob$ is true, the space complexity of \cref{alg:large-universe-sampling} is $O(\sqrt{W} \cdot \text{polylog}(T/\beta)) \cdot \text{poly}(\frac{1}{\rho \eta})$. 
        \item \label{it:space-uncond}If $ob$ is false, the space complexity of \cref{alg:large-universe-sampling} is $O(T^{1/3} \cdot \text{polylog}(T/\beta)) \cdot \text{poly}(\frac{1}{\rho \eta})$. 
\end{enumerate}
\end{theorem}

\begin{proof}

We condition on all the high-probability events used in the proof of \cref{thm:acc_main}; these events occur with probability $1-\beta$ as shown in the proof of \cref{thm:acc_main}. The space usage comes from (1) the KSET data structure in \textsf{COUNTING-KSET}, (2) \textsf{BinaryMechanism-CD}, and (3) the blocklist (when $ob$ is false).
\begin{itemize}
    \item     From \cref{lem:k-set-acc}, the $L$ instantiations of \textsf{COUNTING-KSET} together have space complexity $L \cdot O(k (\log T + \log |\cU|) \log (k/\beta)) = O \left(\log T \cdot k \cdot (\log T + \log |\cU|) \log (k / \beta) \right)$ where $\lambda = 2\log (40 \lceil \log (T) \rceil / \beta )$ and $k = 16 \max \{  \gamma/\eta, 32\lambda / \eta^2 \} +   4\sqrt{2} \frac{(\log T + 1)^{3/2}\sqrt{W \log(20T \lceil \log T \rceil /\beta)}}{\sqrt{\rho}} $ (here $W=T^{2/3}$ when $ob$ is false). We make the standard assumption that the universe size $|\cU|$ is $\text{poly}(T)$, so we can represent an item in $\text{poly}(T)$ bits. This means that the space complexity simplifies to $O(k \cdot \text{poly}\log (T) \cdot \log (k / \beta) )$. 
    
    \item     As shown in \cite{ChanSS2011}, one instance of \textsf{BinaryMechanism-CD} uses $O(\log(T))$ space, and therefore the $L$ copies of the \textsf{BinaryMechanism-CD} use space $O(L\cdot \log(T)) = O(\log^2 (T))$. 

    \item The blocklist when $ob$ is false has size $O(T^{1/3} \log (T/\beta))$ when conditioned on the high-probability events.  
\end{itemize}

The dominating term is the one from the \textsf{KSET} data structures, which is $O(k \cdot \text{polylog} (T) \cdot \log (k / \beta) )$ for $k = 16 \max \{  \gamma/\eta, 32\lambda / \eta^2 \} +   4\sqrt{2} \frac{(\log T + 1)^{3/2}\sqrt{W \log(20T \lceil \log T \rceil /\beta)}}{\sqrt{\rho}} $. 
    
If $ob$ is true, $\gamma=  \sqrt{ \frac{  4  (W+1) (\log T + 1)^3  \log (10 (\log T + 1)/ \beta) }{\rho} }  $ in which case the space complexity simplifies to $O(\sqrt{W} \cdot \text{polylog}(T/\beta)) \cdot \text{poly}(\frac{1}{\rho \eta})$; otherwise, when $ob$ is false, $\gamma =  \sqrt{ \frac{  4  (T^{2/3}+1) (\log T + 1)^3  \log (10 (\log T + 1)/ \beta) }{\rho} }  + 3T^{1/3} \log (5 T^{1/3} \lceil \log T \rceil / \beta) $, in which case the space complexity simplifies to $O(T^{1/3} \cdot \text{polylog}(T/\beta)) \cdot \text{poly}(\frac{1}{\rho \eta})$. 
\end{proof}

\section{Blocklisting Problem: Space Upper and Lower Bounds}
\label{sec:lb}

In this section, we formally define the problems of blocklisting items with high flippancy (and high \occur) and prove a space lower bound for both problems. Our lower bound is information theoretic and applies to any algorithm for blocklisting (flippancy or \occur), including exponential time algorithms and non-private algorithms. 
Then, we show that the problem of blocklisting \occur\ has an almost-matching space upper bound that is tight up to log factors, given by \cref{alg:large-universe-sampling}. 

Recall that the flippancy of an item is defined in~\cite{JainKRSS23} as the number of timesteps where the item switches between being present to absent or vice versa, while \occur\ is defined by the number of timesteps where an item appears in the stream (with any sign).

Informally, we define the $\bf{blocklist}_{flip}(W)$  (resp., $\bf{blocklist}_{occ}(W)$) problem to be the problem of identifying, for each timestep of a turnstile stream, whether the current element of the stream has flipped $<W$ times (resp., occurred $<W$ times) before this timestep. More formally, let $\cU$ be a universe of items and let $x=(x_1,\ldots, x_T)$, be a turnstile stream where for each time $t\in [T]$, the stream element $x_t$ is either an insertion ($+u$) or deletion ($-u$) for some $u \in \cU$, or $x_t=\perp$ indicating an empty update.

\begin{definition}[The $\bf{blocklist}_{flip}(W)$ (resp., $\bf{blocklist}_{occ}(W)$) problem]   
 For the turnstile stream $x=(x_1,\ldots, x_T)$, define the \emph{ground truth} for $x$ as the binary stream of outputs $o^*(x) =(o^*_1, \ldots, o^*_T)$ where for each $t \in [T]$, $o^*_t=0$ if $x_t$ has flippancy $<W$ (resp., \occur\ $< W$) in the prefix stream $(x_1, \ldots, x_{t-1})$ or if $x_t = \perp$, and $o^*_t =1$ otherwise.
 Let $o(x) = (o_1, \ldots, o_T)$ be the output provided by an algorithm on the stream $x$. The algorithm has a \emph{false negative} at time $t$ if $o_t = 0$ and $o^*_t=1$, and has a  \emph{false positive} when $o_t = 1$ and $o^*_t = 0$.
\end{definition}

Notice that algorithms using flippancy blocklisting such as \cite{JainKRSS23}, or 
\occur\ blocklisting like our algorithm are required to have no false negatives with high probability. This is because the max flippancy (and \occur) bound is used to upper bound the sensitivity of the binary tree mechanism and thus needs to hold with probability at least ($1-\delta$) to achieve $(\eps,\delta)$-DP. 

While false negatives affect the privacy of the algorithm, false positives must also be bounded to ensure accuracy. For this reason, we ask if it is possible to design {\it low-space} blocklisting algorithms (for flippancy or \occur) that guarantee (with high probability) no false negatives, while bounding the number of false positives. In Section \ref{s.lowerbound}, we prove a lower bound on the space of any algorithm that bounds flippancy or \occur. Then in Section \ref{s.upperbound}, we give a near-matching upper bound by showing that Algorithm~\ref{alg:large-universe-sampling} solves the low-space \occur\ blocklisting problem using space that matches the lower bound up to log-factors.

\subsection{Lower bound}\label{s.lowerbound}

We first show that any algorithm (including non-private and exponential-time algorithms) for  $\bf{blocklist}_{flip}(W)$ and $\bf{blocklist}_{occ}(W)$, with no false negatives and with bounded false positives must have a space that depends on the number of false positives allowed. 
This space lower bound also extends to any algorithm for count distinct estimation that uses blocklisting methods to control flippancy or \occur, such as \cite{JainKRSS23} or our work.

\begin{restatable}{theorem}{blocklistlb}\label{thm:lb}
For any even integer $W>0$, and any integer $r > 0$,  let $\cA_{flip}$ (resp., $\cA_{occ}$) be an algorithm for $\bf{blocklist}_{flip}(W)$ (resp., $\bf{blocklist}_{occ}(W)$) such that, given an arbitrary stream $x$ of length $T$, with probability at least $1-\beta$, has no false negatives and has at most $r$ false positives.  Then algorithm $\cA_{flip}$ and $\cA_{occ}$ use space at least:
$$(1-\beta) \cdot \left( \log(1-\beta)  + \frac{T}{2W} \log\frac{WT}{T+2Wr}\right).$$
\end{restatable}

We provide a proof sketch here, 
and a full proof is deferred to Appendix \ref{s.appprooflb}. 
Our lower bound proceeds by describing a random process defining a stream distribution that is hard for the problems of $\bf{blocklist}_{flip}(W)$ and $\bf{blocklist}_{occ}(W)$. For the element universe $\cU = [T/2]$, the main idea is to define a distribution of problem instances as follows. First, generate  $X \subseteq \cU$ as a uniformly random set of size $T/2W$. Then define a stream $x(X)$ where for the first $T/2$ timesteps, each element $u \in X$ appears in $W$ updates, alternating $W/2$ times between one insertion and one deletion of $u$. This results in $W$ flippancy and $W$ \occur\ for all elements of $X$ at the end of the first half of the stream. In the second $T/2$ timesteps of the stream,  all elements in $\cU$ are inserted once. Thus the correct output for this stream for both the $\bf{blocklist}_{flip}(W)$ and $\bf{blocklist}_{occ}(W)$ problems is to always output 0 in the first half, and to output 1 in the second half only for elements in $X$.

We then use an information theory argument to show a space lower bound for any algorithm $\cA_{flip}$ that satisfies the conditions of no false negatives and at most $r$ false positives, with probability at least $1-\beta$ over this distribution of problem instances.

\subsection{Upper bound for \occur\ blocklisting}\label{s.upperbound}

We now present an upper bound (Corollary \ref{corol:blocklist-ub}) based on Algorithm~\ref{alg:large-universe-sampling} for the space required to solve the \occur\ blocklisting problem with $r$ false positives.

\begin{restatable}{corollary}{blocklistub}\label{corol:blocklist-ub}
    With probability $1-2\beta$ and when $|\cU|=poly(T)$, Algorithm~\ref{alg:large-universe-sampling} with $ob$=false reports no false negatives and $ r= 2T^{1/3} \log (T^{1/3} \lceil \log(T) \rceil /\beta)$ false positives for the problem $\bf{blocklist}_{occ}(W)$  for $W=T^{2/3}$, while using space $O(T^{1/3} \cdot \text{polylog}(T/\beta)) \cdot \text{poly}(\frac{1}{\rho \eta}))$. 
\end{restatable}

Note that plugging $ r= 2T^{1/3} \log (T^{1/3} \lceil \log(T) \rceil /\beta)$ and $W=T^{2/3}$ into the lower bound of \cref{thm:lb} gives an near-matching space lower bound that is tight up to log factors. 

The proofs for no false negatives and the space complexity follow from \cref{lem:privacy_occ_bound} in \cref{app.exactcountproof}, which gives a high probability bound on the maximum \occur\ of the stream after blocklisting, and from \cref{thm:space}, which bounds the space used by Algorithm~\ref{alg:large-universe-sampling}. The proof for the bounded number of false positives follows from a concentration bound the probability of blocklisting an element too early. The full proof is deferred to Appendix \ref{s.appproofub}.

\section{Conclusions}
In this paper we designed the first space-efficient differentially private algorithms for the count distinct element problem in the turnstile model. This result addresses an open question of~\cite{JainKRSS23}, showing that it is possible to design a low memory DP algorithm for this problem in the turnstile setting. While we show that any algorithm that uses blocklisting techniques cannot do any better in terms of space, an interesting open question is to prove unconditional space bounds for any DP continual release algorithm addressing the problem (regardless of the techniques used). Currently the theoretical understanding of space lower bounds in the DP streaming setting is very limited. Only recently~\cite{DinurSWZ23} provided the first space DP lower bound for any problem, under cryptographic assumptions; any future progress in this direction would be interesting.

\bibliographystyle{alpha}
\bibliography{references}

\newcommand{\etalchar}[1]{$^{#1}$}
\begin{thebibliography}{MMNW11}

\bibitem[AMS96]{alon1996space}
Noga Alon, Yossi Matias, and Mario Szegedy.
\newblock The space complexity of approximating the frequency moments.
\newblock In {\em Proceedings of the 28th ACM Symposium on Theory of
  Computing}, STOC `96, pages 20--29, 1996.

\bibitem[AMS99]{AlonMS99}
Noga Alon, Yossi Matias, and Mario Szegedy.
\newblock The space complexity of approximating the frequency moments.
\newblock {\em Journal of Computer and System Sciences}, 58(1):137--147, 1999.

\bibitem[BBDS12]{BlockiBDS12}
Jeremiah Blocki, Avrim Blum, Anupam Datta, and Or~Sheffet.
\newblock The {J}ohnson-{L}indenstrauss {T}ransform {I}tself {P}reserves
  {D}ifferential {P}rivacy.
\newblock In {\em 53rd Annual {IEEE} Symposium on Foundations of Computer
  Science}, FOCS `12, pages 410--419, 2012.

\bibitem[BFM{\etalchar{+}}13]{BolotFMNT13}
Jean Bolot, Nadia Fawaz, S.~Muthukrishnan, Aleksandar Nikolov, and Nina Taft.
\newblock Private decayed predicate sums on streams.
\newblock In {\em Proceedings of the 16th International Conference on Database
  Theory}, ICDT `13, pages 284--295, 2013.

\bibitem[BGMZ23]{BlockiGMZ23}
Jeremiah Blocki, Elena Grigorescu, Tamalika Mukherjee, and Samson Zhou.
\newblock How to make your approximation algorithm private: {A} black-box
  differentially-private transformation for tunable approximation algorithms of
  functions with low lensitivity.
\newblock In {\em Approximation, Randomization, and Combinatorial Optimization.
  Algorithms and Techniques, {APPROX/RANDOM} 2023}, volume 275 of {\em LIPIcs},
  pages 59:1--59:24, 2023.

\bibitem[BMWZ23]{BravermanMWZ23}
Vladimir Braverman, Joel Manning, Zhiwei~Steven Wu, and Samson Zhou.
\newblock {P}rivate {D}ata {S}tream {A}nalysis for {U}niversal {S}ymmetric
  {N}orm {E}stimation.
\newblock In {\em Approximation, Randomization, and Combinatorial Optimization.
  Algorithms and Techniques, {APPROX/RANDOM} 2023}, volume 275 of {\em LIPIcs},
  pages 45:1--45:24, 2023.

\bibitem[BR94]{BR94}
M.~Bellare and J.~Rompel.
\newblock Randomness-efficient oblivious sampling.
\newblock In {\em Proceedings of the 35th Annual Symposium on Foundations of
  Computer Science}, FOCS `94, pages 276--287, 1994.

\bibitem[BS16]{BunS16}
Mark Bun and Thomas Steinke.
\newblock Concentrated differential privacy: simplifications, extensions, and
  lower bounds.
\newblock In {\em Proceedings of the 14th International Conference on Theory of
  Cryptography}, TCC `16, pages 635--658, 2016.

\bibitem[Cha12]{chakrabarti-streaming}
Amit Chakrabarti.
\newblock Lecture notes on data stream algorithms, 2012.
\newblock Available at
  \url{https://www.cs.dartmouth.edu/~ac/Teach/data-streams-lecnotes.pdf}.

\bibitem[CM05]{cormode2005improved}
Graham Cormode and Shan Muthukrishnan.
\newblock An improved data stream summary: The count-min sketch and its
  applications.
\newblock {\em Journal of Algorithms}, 55(1):58--75, 2005.

\bibitem[COP03]{charikar2003better}
Moses Charikar, Liadan O'Callaghan, and Rina Panigrahy.
\newblock Better streaming algorithms for clustering problems.
\newblock In {\em Proceedings of the 35th Annual ACM Symposium on Theory of
  Computing}, STOC `03, pages 30--39, 2003.

\bibitem[CSS11]{ChanSS2011}
T.-H.~Hubert Chan, Elaine Shi, and Dawn Song.
\newblock Private and continual release of statistics.
\newblock {\em ACM Transactions on Information and System Security}, 14(3),
  2011.

\bibitem[DF03]{durand2003loglog}
Marianne Durand and Philippe Flajolet.
\newblock Loglog counting of large cardinalities.
\newblock In {\em Proceedings of the European Symposium on Algorithms}, ESA
  `03, pages 605--617, 2003.

\bibitem[DGIM02]{datar2002maintaining}
Mayur Datar, Aristides Gionis, Piotr Indyk, and Rajeev Motwani.
\newblock Maintaining stream statistics over sliding windows.
\newblock {\em SIAM Journal on Computing}, 31(6):1794--1813, 2002.

\bibitem[DLB19]{DesfontainesLB19}
Damien Desfontaines, Andreas Lochbihler, and David~A. Basin.
\newblock Cardinality estimators do not preserve privacy.
\newblock In {\em Proceedings on Privacy Enhancing Technologies}, volume 2019
  of {\em PETS `19}, 2019.

\bibitem[DMNS06]{dwork2006calibrating}
Cynthia Dwork, Frank McSherry, Kobbi Nissim, and Adam Smith.
\newblock Calibrating noise to sensitivity in private data analysis.
\newblock In {\em Proceedings of 3rd Theory of Cryptography Conference}, TCC
  '06, pages 265--284, 2006.

\bibitem[DMP{\etalchar{+}}24]{mcmahan2024efficient}
Krishnamurthy Dvijotham, H~Brendan McMahan, Krishna Pillutla, Thomas Steinke,
  and Abhradeep Thakurta.
\newblock Efficient and near-optimal noise generation for streaming
  differential privacy.
\newblock In {\em 65th IEEE Symposium on Foundations of Computer Science}, FOCS
  `24, 2024.

\bibitem[DNPR10]{DworkNPR10}
Cynthia Dwork, Moni Naor, Toniann Pitassi, and Guy~N. Rothblum.
\newblock Differential privacy under continual observation.
\newblock In {\em Proceedings of the 42nd {ACM} Symposium on Theory of
  Computing}, STOC `10, pages 715--724, 2010.

\bibitem[DR14]{dwork2014algorithmic}
Cynthia Dwork and Aaron Roth.
\newblock The algorithmic foundations of differential privacy.
\newblock {\em Foundations and Trends in Theoretical Computer Science},
  9(3--4):211--407, 2014.

\bibitem[DSWZ23]{DinurSWZ23}
Itai Dinur, Uri Stemmer, David~P. Woodruff, and Samson Zhou.
\newblock On differential privacy and adaptive data analysis with bounded
  space.
\newblock In {\em Proceedings of the 42nd Annual International Conference on
  the Theory and Applications of Cryptographic Techniques}, EUROCRYPT `23,
  pages 35--65, 2023.

\bibitem[DTT22]{DickensTT22}
Charlie Dickens, Justin Thaler, and Daniel Ting.
\newblock Order-invariant cardinality estimators are differentially private.
\newblock In {\em Advances in Neural Information Processing Systems 35: Annual
  Conference on Neural Information Processing Systems}, NeurIPS `22, 2022.

\bibitem[EMM{\etalchar{+}}23]{EpastoMMMVZ23}
Alessandro Epasto, Jieming Mao, Andres~Mu{\~{n}}oz Medina, Vahab Mirrokni,
  Sergei Vassilvitskii, and Peilin Zhong.
\newblock Differentially private continual releases of streaming frequency
  moment estimations.
\newblock In {\em 14th Innovations in Theoretical Computer Science Conference},
  volume 251 of {\em ITCS `23}, pages 48:1--48:24, 2023.

\bibitem[FFGM12]{flajolet2007hyperloglog}
Philippe Flajolet, {\'E}ric Fusy, Olivier Gandouet, and Fr{\'e}d{\'e}ric
  Meunier.
\newblock Hyperloglog: the analysis of a near-optimal cardinality estimation
  algorithm.
\newblock {\em Discrete Mathematics \& Theoretical Computer Science}, 2012.

\bibitem[FHO21]{FichtenbergerHO21}
Hendrik Fichtenberger, Monika Henzinger, and Lara Ost.
\newblock Differentially private algorithms for graphs under continual
  observation.
\newblock In {\em 29th Annual European Symposium on Algorithms, {ESA}}, volume
  204 of {\em LIPIcs}, pages 42:1--42:16. Schloss Dagstuhl - Leibniz-Zentrum
  f{\"{u}}r Informatik, 2021.

\bibitem[FM85]{flajolet1985probabilistic}
Philippe Flajolet and G~Nigel Martin.
\newblock Probabilistic counting algorithms for data base applications.
\newblock {\em Journal of computer and system sciences}, 31(2):182--209, 1985.

\bibitem[Gan07]{Ganguly07}
Sumit Ganguly.
\newblock Counting distinct items over update streams.
\newblock {\em Theoretical Computer Science}, 378(3):211--222, 2007.

\bibitem[GKNM23]{Ghazi0NM23}
Badih Ghazi, Ravi Kumar, Jelani Nelson, and Pasin Manurangsi.
\newblock Private counting of distinct and k-occurring items in time windows.
\newblock In {\em 14th Innovations in Theoretical Computer Science Conference},
  ITCS `23, pages 55:1--55:24, 2023.

\bibitem[HP19]{HuangP19}
Zengfeng Huang and Pan Peng.
\newblock Dynamic graph stream algorithms in o(n) space.
\newblock {\em Algorithmica}, 81(5):1965--1987, 2019.

\bibitem[HSS23]{HenzingerSS23}
Monika Henzinger, A.~R. Sricharan, and Teresa~Anna Steiner.
\newblock Differentially private data structures under continual observation
  for histograms and related queries, 2023.
\newblock arXiv pre-print 2302.11341.

\bibitem[HTC23]{HehirTC23}
Jonathan Hehir, Daniel Ting, and Graham Cormode.
\newblock Sketch-{F}lip-{M}erge: {M}ergeable sketches for private distinct
  counting.
\newblock In {\em Proceedings of the International Conference on Machine
  Learning}, 2023.

\bibitem[JKR{\etalchar{+}}23]{JainKRSS23}
Palak Jain, Iden Kalemaj, Sofya Raskhodnikova, Satchit Sivakumar, and Adam~D.
  Smith.
\newblock Counting distinct elements in the turnstile model with differential
  privacy under continual observation.
\newblock In {\em Advances in Neural Information Processing Systems 36: Annual
  Conference on Neural Information Processing Systems}, NeurIPS `23, 2023.

\bibitem[MG82]{misra1982finding}
Jayadev Misra and David Gries.
\newblock Finding repeated elements.
\newblock {\em Science of Computer Programming}, 2(2):143--152, 1982.

\bibitem[MMNW11]{MirMNW11}
Darakhshan~J. Mir, S.~Muthukrishnan, Aleksandar Nikolov, and Rebecca~N. Wright.
\newblock {P}an-private algorithms via statistics on sketches.
\newblock In {\em Proceedings of the 30th {ACM} {SIGMOD-SIGACT-SIGART}
  Symposium on Principles of Database Systems}, PODS `11, pages 37--48, 2011.

\bibitem[Mor78]{morris1978counting}
Robert Morris.
\newblock Counting large numbers of events in small registers.
\newblock {\em Communications of the ACM}, 21(10):840--842, 1978.

\bibitem[MU17]{mitzenmacher2017probability}
Michael Mitzenmacher and Eli Upfal.
\newblock {\em Probability and computing: Randomization and probabilistic
  techniques in algorithms and data analysis}.
\newblock Cambridge university press, 2017.

\bibitem[Mut05]{muthukrishnan2005data}
Shanmugavelayutham Muthukrishnan.
\newblock Data streams: Algorithms and applications.
\newblock {\em Foundations and Trends in Theoretical Computer Science},
  1(2):117--236, 2005.

\bibitem[MV18]{mcgregor2018simple}
Andrew McGregor and Sofya Vorotnikova.
\newblock A simple, space-efficient, streaming algorithm for matchings in low
  arboricity graphs.
\newblock In {\em Proceedings of 1st Symposium on Simplicity in Algorithms},
  SOSA `18, 2018.

\bibitem[PS21]{PaghMS2020}
Rasmus Pagh and Nina~Mesing Stausholm.
\newblock Efficient differentially private {F}\({}_{\mbox{0}}\) linear
  sketching.
\newblock In {\em Proceedings of the 24th International Conference on Database
  Theory}, ICDT `21, pages 18:1--18:19, 2021.

\bibitem[SNY17]{StanojevicNY17}
Rade Stanojevic, Mohamed Nabeel, and Ting Yu.
\newblock Distributed cardinality estimation of set operations with
  differential privacy.
\newblock In {\em Proceedings of the {IEEE} Symposium on Privacy-Aware
  Computing}, PAC `17, pages 37--48, 2017.

\bibitem[SST20]{Smith0T20}
Adam~D. Smith, Shuang Song, and Abhradeep Thakurta.
\newblock The {F}lajolet-{M}artin sketch itself preserves differential privacy:
  {P}rivate counting with minimal space.
\newblock In {\em Advances in Neural Information Processing Systems 33},
  NeurIPS `20, pages 19561--19572, 2020.

\bibitem[TS13]{TS13}
Abhradeep~Guha Thakurta and Adam Smith.
\newblock ({N}early) optimal algorithms for private online learning in
  full-information and bandit settings.
\newblock In {\em Advances in Neural Information Processing Systems 26},
  volume~26 of {\em NeurIPS `13}, 2013.

\bibitem[Woo04]{woodruff2004optimal}
David~P Woodruff.
\newblock Optimal space lower bounds for all frequency moments.
\newblock In {\em Proceedings of the 15th annual ACM-SIAM Symposium on Discrete
  Algorithms}, SODA `04, pages 167--175, 2004.

\bibitem[WPS22]{WangPS22}
Lun Wang, Iosif Pinelis, and Dawn Song.
\newblock Differentially private fractional frequency moments estimation with
  polylogarithmic space.
\newblock In {\em Proceedings of the 10th International Conference on Learning
  Representations}, ICLR `22, 2022.

\end{thebibliography}

\clearpage
\appendix

\section{Additional Tools}

\subsection{Concentration bounds} We provide some basic concentration inequalities that will be used in our analyses. 

\begin{lemma}[\cite{BR94}]\label{lem:hash}
Let $\lambda \geq 4$ be an even integer. Let $X$ be the sum of $n$ $\lambda$-wise independent random variables which take values in $[0,1]$. Let $\mu = E[X] $ and $A > 0$. Then,
$$\Pr [|X - \mu| > A] \leq 8 \left( \frac{\lambda \mu + \lambda^2}{A^2}\right)^{\lambda / 2}.$$
\end{lemma}

\begin{theorem}[Multiplicative Chernoff Bound \cite{mitzenmacher2017probability}]\label{thm:mult_chernoff}
    Let $X = \sum_{i = 1}^n X_i$ where each $X_i$ is a Bernoulli variable which takes value $1$ with probability $p_i$ and value $0$
    with probability $1-p_i$. Let $\mu = \E[X] = \sum_{i = 1}^n p_i$. Then, 
    \begin{enumerate}
        \item Upper Tail: $\Pr[X \geq (1+\eta) \cdot \mu] \leq \exp\left(-\frac{\eta^2\mu}{2 + \eta}\right)$ for all $\eta > 0$;
        \item Lower Tail: $\Pr[X \leq (1-\eta) \cdot \mu] \leq \exp\left(-\frac{\eta^2\mu}{3}\right)$ for all $0 < \eta < 1$.
    \end{enumerate}
\end{theorem}

\begin{lemma}[Chernoff Bound of Gaussian Random Variable]\label{lem:chernoff}  For $X \sim \N(0, \sigma^2)$, $\Pr(|X| > t) \leq 2\exp (-t^2 / 2\sigma^2) $.
\end{lemma}

\subsection{Information theory basics}
\label{sec:info}
We provide some basic information theory definitions and facts that are used in Section \ref{sec:lb}. In this paper, we use $\log$ to refer to the base $2$ logarithm.

\begin{definition}
The \emph{entropy} of a random variable $X$, denoted by $H(X)$, is defined as $H(X) = \sum_x \Pr[X = x] \log(1 / \Pr[X = x])$. 
\end{definition}

\begin{definition}
The \emph{conditional entropy} of random variable $X$ conditioned on random variable $Y$ is defined as $H(X|Y) = \mathbb{E}_y[H(X|Y = y)] = \sum_y \Pr[Y=y] \cdot H(X|Y=y)$. 
\end{definition}
\begin{definition}
\label{def:muinfo}
The \emph{mutual information} between two random variables $X$ and $Y$ is defined as $I(X;Y) = H(X) - H(X|Y) = H(Y) - H(Y|X)$. 
\end{definition}

\begin{definition}
The \emph{conditional mutual information} between $X$ and $Y$ given $Z$ is defined as $I(X;Y|Z) =  \mathbb{E}_z[I(X; Y|Z = z)]$
\end{definition}

\begin{fact}
\label{fact:infofact}
Let $X,Y,Z$ be three random variables.
\begin{enumerate}
\item $H(X|Y) \geq H(X|Y,Z)$.
\item $H(X) \leq \log | \text{supp}(X) |$.
\item $I(X;Y|Z) \leq H(X|Z)$.
\item Data processing inequality: for a deterministic function $f(X)$, $I(X;Y|Z) \geq I(f(X);Y|Z)$.
\item $I(X;Y|Z) \geq 0$.
\end{enumerate}
\end{fact}

\section{Additional details on \textsf{KSET}}\label{app:kset}
We describe the \textsf{TESTSINGLETON} data structure (Algorithm~\ref{alg:ts}) which is a building block of the \textsf{KSET} data structure in more detail.

\begin{algorithm}[!ht]
\caption{TEST-SINGLETON data structure }
\label{alg:ts}
\begin{algorithmic}[1]
\REQUIRE{Input stream $x_1, x_2, \ldots, x_T$}
\STATE{Initialize $m \rightarrow 0, U \rightarrow 0, V \rightarrow 0$}
\STATE{\textbf{TSUPDATE$(x_t)$:}} 
\IF{$x_t$ is an insertion of item $i$}
\STATE{$m \leftarrow m + 1, U \rightarrow U + i , V \rightarrow V + i^2$}
\ELSIF{$x_t$ is a deletion}
\STATE{$m \leftarrow m - 1, U \rightarrow U - i , V \rightarrow V - i^2$}
\ENDIF
\STATE{\textbf{TSCARD():}}
\IF{$m = 0$}
\STATE{Return EMPTY}
\ELSIF{$U^2  = m \cdot V$}
\STATE{Return (SINGLETON, $U/m$, $m$)}
\ELSE
\STATE{Return COLLISION}
\ENDIF
\end{algorithmic}
\end{algorithm}

The \textsf{TESTSINGLETON} data structure supports the following operations:
\begin{enumerate}
    \item An update operation, \textsf{TSUPDATE}$(x_t)$, which updates  three counters --- $m_{TS}$, $U$, and $V$ (all initialized to zero) preserving the following invariants throughout the stream:
    \begin{align*}
        m_{TS}=\sum_{a \in \cU} f_a,\; \; U=\sum_{a \in \cU} f_a\cdot a,\;\; V=\sum_{a \in \cU} f_a\cdot a^2.
    \end{align*}
    More precisely, for an non-empty update $x_t$ corresponding to data item $a$, \textsf{TSUPDATE}$(x_t)$ performs the following update:
    \begin{align*}
        m_{TS}&:= m_{TS}+1,\; \; U=U+a,\; \; V=V+a^2 \; \; \\
        &\textrm{(for an addition)}\\
        m_{TS}&:= m_{TS}-1,\; \; U=U-a,\; \; V=V-a^2 \; \; \\
        &\textrm{(for a deletion)}.
    \end{align*}
    \item A check operation, \textsf{TSCARD}$()$, which determines whether the \textsf{TESTSINGLETON} data structure: (1) is empty, (2) contains a single element, or (3) has more than a single element. The function returns, in each case respectively: (1) EMPTY (this happens if $m_{TS}=0$); (2) the triplet SINGLETON, the element, and its frequency (this happens if $U^2=m_{TS}\cdot V)$; or (3) COLLISION (if the last two checks fail). It is easy to see that the unique item returned (in the SINGLETON case) has identity $\frac{U}{m_{TS}}$ and has frequency $m_{TS}$.
\end{enumerate}

\section{Omitted Proofs from Section \ref{sec:analysis-cd}}\label{app.proofs}

For ease of analysis, we define the streams produced after applying the hash function and the blocklisting procedure (in the case when $ob=false$) as follows. These streams will be used as intermediate steps in the analysis of \textsf{CountDistinct} to separately reason about the hashing and blocklisting procedures, and their impact on sensitivity of the resulting streams.

\begin{definition}\label{def:hash-stream}
    Define $\cS_{i,g}$ as the substream of $x$ after applying hash function $g$. That is, let $a$ be the item contained in the update $x_t$. Then $\cS_{i,g}[t]=x_t$ if $g(a)=i$ and $\cS_{i,g}[t]=\bot$ otherwise.  
\end{definition}

\begin{definition}\label{def:bl-stream}
    Define $\cS_{i,B}$ as the stream of updates produced from $\cS_{i,g}$
     after checking whether the item corresponding to the update $x_t$ is in the blocklist $\cB$ before time $t$ or not. That is, if the item corresponding to $x_t$ is in $\cB$ before time $t$, then $\cS_{i,B}[t]=\bot$, otherwise $\cS_{i,B}[t]=\cS_{i,g}[t]$. 
\end{definition}

\subsection{Proof of Lemma \ref{lem:new_privacy} and Helper Lemma}\label{app.newprivacyproof}

\ksetexactmatch*

\begin{proof}
We start with the case when $ob=false$, which is the more involved case. 
Consider the randomness of \textsf{COUNTING-KSET}$_i$ and \textsf{COUNTING-DICT}$_i$. Observe that because of the fixed randomness of both hashing and blocklisting, the resulting streams $\cS_{i,B}$ (see \cref{def:bl-stream}) that are respectively used to update the \textsf{KSET} (in the case of \textsf{COUNTING-KSET}$_i$) and the \textsf{DICT} data structure (in the case of \textsf{COUNTING-DICT}$_i$) are identical.

Next, define the randomness of \textsf{COUNTING-KSET}$_i$ as $\cR_{KC_i} = \cR_{KS_i}\times \cR_{BM_i}$ where $\cR_{KS_i}$ denotes the randomness from the \textsf{KSET} (\cref{alg:kset}) and $\cR_{BM_i}$ denotes the randomness from \textsf{BinaryMechanism-CD} (\cref{alg:bm}).
On the other hand, the only randomness in \textsf{COUNTING-DICT}$_i$ is due to the randomness of \textsf{BinaryMechanism-CD}, i.e., $\cR_{EC_i} = \cR_{BM_i}$. We emphasize that because the randomness from adding items to the blocklist has been fixed across \textsf{CountDistinct} and \textsf{CountDistinct'}, the blocklist $\cB$ passed to both \textsf{COUNTING-KSET}$_i$ and \textsf{COUNTING-DICT}$_i$ are identical. 

We now want to argue that the outputs of \textsf{COUNTING-KSET}$_i$ and \textsf{COUNTING-DICT}$_i$ are the same except with probability $\beta$. Let the \textsf{BinaryMechanism-CD} instance in \textsf{COUNTING-KSET}$_i$ be denoted as \textsf{BinaryMechanism-CD}$_{KC_i}$ and the \textsf{BinaryMechanism-CD} instance in \textsf{COUNTING-DICT}$_i$ as \textsf{BinaryMechanism-CD}$_{EC_i}$, and fix the randomness used in \textsf{BinaryMechanism-CD}$_{KC_i}$ and \textsf{BinaryMechanism-CD}$_{EC_i}$. 

We claim that there are two bad events for which the outputs of \textsf{COUNTING-KSET}$_i$ and \textsf{COUNTING-DICT}$_i$ may differ. 

\begin{itemize}
  \item \textbf{E}$_1$: There exists a timestep where the true count $\leq k$ and 
    the \textsf{KSET} outputs NIL. 
   \item \textbf{E}$_2$: There exists a timestep where the true count is $> k$ and the noisy count of \textsf{COUNTING-DICT}$_i \leq k$.   
\end{itemize}

By setting $k \geq \tau + O\left( \frac{\text{polylog}  (T/\beta) \sqrt{W} }{\sqrt{\rho}} \right)$, we argue that both events happen with probability at most $\beta$. 

For the event \textbf{E}$_1$, the probability of the \textsf{KSET} outputting NIL happens with probability at most $\beta/2TL$ for one timestep. This is because, in Line~\ref{line:kset-count-init} of \textsf{COUNTING-KSET}$_i$, we set $k$ as the capacity of the \textsf{KSET} and $\beta/2TL$ as the failure probability. So by \cref{it:kset-s} in \cref{lem:k-set-acc} and union bound over all timesteps, this event happens with probability at most $\beta/2L$. 

For the event \textbf{E}$_2$, \cref{lem:e2-occur} (below) bounds the probability of $\textbf{E}_2$ as $\beta/2L$ over all timesteps for our choice of $k$.

Now, conditioned on bad events \textbf{E}$_1$ and \textbf{E}$_2$ \emph{not} occurring over all timesteps, we argue that the outputs of \textsf{COUNTING-KSET}$_i$ and \textsf{COUNTING-DICT}$_i$ are identical. Let $t_1$ be the \emph{first} timestep that the true count is $>k$. Let $t_2 > t_1$ be the next timestep that the true count is $\leq k$. We next consider the outputs of the two algorithms by cases across timesteps.
\begin{itemize}
    \item \textbf{Case 1: $1 \leq t < t_1$.} Conditioned on event \textbf{E}$_1$ not occurring, the \textsf{KSET} does not output NIL during this time epoch, which means that the inputs to \textsf{BinaryMechanism-CD}$_{KC_i}$ and \textsf{BinaryMechanism-CD}$_{EC_i}$ are identical, since both \textsf{BinaryMechanism-CD}$_{KC_i}$ and \textsf{BinaryMechanism-CD}$_{EC_i}$ will be updated with only the update from the previous timestep. In this case, the resulting noisy outputs will be the same under the fixed randomness, and the output of \textsf{COUNTING-KSET}$_i$ and \textsf{COUNTING-DICT}$_i$ after the thresholding step (comparison to $\tau$) is identical. 

\item \textbf{Case 2: $t_1 \leq  t < t_2 $.} 
For timesteps in this epoch, both \textsf{COUNTING-KSET}$_i$ and \textsf{COUNTING-DICT}$_i$ will output TOO-HIGH. Since the true count is $>k$ for all $t_{1} \leq t < t_{2} $, and conditioned on \textbf{E}$_2$ not occurring, the output of \textsf{COUNTING-DICT}$_i$ must be TOO-HIGH. Also by \cref{it:kset-full} in \cref{lem:k-set-acc}, the \textsf{KSET} outputs NIL for this time period with probability $1$, which means that the output of \textsf{COUNTING-KSET}$_i$ is also TOO-HIGH. 
 
\item \textbf{Case 3: $t = t_2$.} 
Next we argue about the output of \textsf{COUNTING-DICT}$_i$ and \textsf{COUNTING-KSET}$_i$ at timestep $t_2$ when the true count $\leq k$.  Conditioning on event \textbf{E}$_1$ not occurring, the \textsf{KSET} does not output NIL at timestep $t_2$ because the true count $\leq k$. Moreover, observe that \textsf{BinaryMechanism-CD}$_{KC_i}$ is not updated over $t_1< t <t_2$ and is only updated at timestep $t_2$ because the \textsf{KSET} does not output NIL. Also, by construction of \textsf{COUNTING-KSET}$_i$, we claim that \textsf{BinaryMechanism-CD}$_{KC_i}$ is fed a sequence of inputs ($+1,-1,0$) at timestep $t_{2}$ that result in the same sum and the same length as in \textsf{BinaryMechanism-CD}$_{EC_i}$ over $t_{2} - t_{1}$ timesteps. This is because by definition, $\textsf{diff}$ is the difference in the number of distinct elements between times $t_1-1$ and $t_2$, and $\vert \textsf{diff} \vert \leq t_{\textsf{diff}}$, as there can only be $\leq t_{\textsf{diff}}$ many distinct elements added (or removed) over $t_2-t_1$.  Since the length and sum of the sequence of inputs to both \textsf{BinaryMechanism-CD}$_{KC_i}$ and \textsf{BinaryMechanism-CD}$_{EC_i}$ at timestep $t_{2}$ is the same, the outputs of both \textsf{COUNTING-KSET}$_i$ and \textsf{COUNTING-DICT}$_i$ are the same at timestep $t_{2}$ under the fixed randomness. 
\end{itemize}
   
   This argument can be extended over all timesteps by iteratively considering the \emph{next} timestep when the true count is $>k$ and the following timestep when the true count is $\leq k$. Thus when $ob=true$, except with probability $\beta/L$ corresponding to the events \textbf{E}$_1$ and \textbf{E}$_2$ occurring, \textsf{COUNTING-KSET}$_i$ and \textsf{COUNTING-DICT}$_i$ will produce identical outputs at each timestep. That is, the distributions of \textsf{COUNTING-KSET}$_i$ and \textsf{COUNTING-DICT}$_i$, denoted $K$ and $E$ respectively, will agree on all outcomes except a subset of probability mass $\beta/L$, which implies that that $d_{TV}(K,E)\leq \beta/L$. 

   For the case when $ob=true$, the blocklisting step is not needed. Then the fixed randomness between \textsf{CountDistinct} and \textsf{CountDistinct}' means that the randomness of the hash functions of both algorithms will be the same, so the resulting stream $\cS_{i,g}$ (see \cref{def:hash-stream}) that is used to update the \textsf{KSET} (in the case of \textsf{COUNTING-KSET}$_i$) and the \textsf{DICT} data structure (in the case of \textsf{COUNTING-DICT}$_i$) is identical. The rest of the argument follows symmetrically to the case when $ob=false$. 
\end{proof}

\subsubsection{Helper Lemma}

\begin{lemma}\label{lem:e2-occur}
Let \textbf{E}$_{BM}$ be the event that there exists a timestep where the noisy count of \textsf{BinaryMechanism-CD}$_{EC_i}$ is greater than $k$ when the true count is less than $\tau$, or the noisy count of \textsf{BinaryMechanism-CD}$_{EC_i}$ is less than $\tau$ when the true count is greater than $k$. The probability of \textbf{E}$_{BM}$ is at most $\beta/2L$ over all timesteps when $k \geq \tau + 2\sqrt{2} \frac{(\log T + 1)^{3/2}\sqrt{W \log(4T\lceil \log(T) \rceil/\beta)}}{\sqrt{\rho}} $. 
\end{lemma}
\begin{proof}
For notational convenience, let $\Delta = k - \tau$. 
From Algorithm~\ref{alg:bm}, we know that the noise that we apply to the true count is a summation of at most $m \leq \log T + 1$ Gaussian random variables, each sampled from $\N(0, 4W(\log T + 1) L / \rho)$. (Recall that the input privacy parameter to the binary mechanism is $\rho/L$). Thus the overall noise added is $N(0, 4mW(\log T + 1) L / \rho)$. 
Now we want to bound the probability that $|\N(0, 4mW(\log T + 1) L / \rho)| > \Delta$, which is an upper bound on the probability of $\textbf{E}_{BM}$ occurring at a single timestep. 

Applying a Chernoff bound (Lemma~\ref{lem:chernoff}) yields:
\begin{align*}
    & \Pr(|\N(0, 4mW(\log T + 1) L / \rho) | > \Delta) \leq  2\exp (-\frac{\Delta^2 \rho}{8mW(\log T+1) L} )
\end{align*}
We wish to bound the above term on the right by $\beta/2TL$, so that then by a union bound, the probability of $\textbf{E}_{BM}$ over all timesteps is bounded by $\beta/2L$. This requires:
\begin{align*}
- \frac{\Delta^2 \rho} {8mWL(\log T+ 1)} &\leq \log (\beta/4TL)\\
\Longleftrightarrow \quad \quad  \Delta  &\geq 2\sqrt{2} \frac{\sqrt{mWL(\log T+1) \log(4TL/\beta)}}{\sqrt{\rho}}
\end{align*}
Recall that our goal is to set the value of $\Delta = k -\tau$ such that the above inequality always holds, and we want to set $\Delta$ to be the upper bound of the right hand side. 
Since $m \leq \log T + 1$  and $L = \lceil \log T \rceil \leq \log T + 1$, 
then, choosing $k \geq t + 2\sqrt{2} \frac{(\log T + 1)^{3/2}\sqrt{W \log(4T\lceil \log(T)\rceil/\beta)}}{\sqrt{\rho}} $ will ensure that $\Delta$ is a valid upper bound, and hence that the probability of $\textbf{E}_{BM}$ over all timesteps is bounded by $\beta / 2L$. 
\end{proof}


\subsection{Proof of Lemma \ref{lem:exact_count_dp} and Helper Lemmas}\label{app.exactcountproof}

\exactcountdp*

\begin{proof}
We will prove the privacy claim for the more general case when $ob=false$. Note that when $ob=true$, we do not need to deal with the failure event associated with blocklisting (\cref{lem:privacy_occ_bound}) and thus $\beta=0$ and \textsf{COUNTING-DICT}$_i$ (\cref{alg:exact-count}) is $\rho/L$-zCDP. 

The key point we must show is that when neighboring streams are input to \textsf{COUNTING-DICT}$_i$, then the internal streams passed to \textsf{BinaryMechanism-CD}$_i$ inside of \textsf{COUNTING-DICT}$_i$ will remain neighboring. Once this is shown, then we can directly apply \cref{lem:bm-ec-dp}, which shows that this instance of \textsf{BinaryMechanism-CD}$_i$ inside \textsf{COUNTING-DICT}$_i$ is differentially private. Thus we must show that even after applying the hashing and blocklisting operations to the original neighboring input streams, the resulting processed streams remain neighboring.

  The randomness of \textsf{CountDistinct'} can be viewed as a joint probability distribution $\cR_{CD'}= \cR_{g} \times \cR_{BL} \times \cR_{EC_1}\times \ldots \times \cR_{EC_L}$ where $\cR_g$ denotes the randomness from picking a hash function $g$ (in Line \ref{line:hash} of \cref{alg:large-universe-sampling}), $\cR_{BL}$ denotes the randomness from blocklisting,  and $\cR_{EC_i}$ denotes the randomness from the subroutine \textsf{COUNTING-DICT}$_i$ for $i\in [L]$.   
   Let $x$ and $x'$ be neighboring streams that differ only at timestep $t^*$, in which the update (either deletion or addition) in $x$ is for item $u$, and in $x'$ is $\bot$, and fix the randomness used in \textsf{CountDistinct}' across runs on $x$ and $x'$.

Let $\cS_{i,g}$ and $\cS'_{i,g}$ be the substreams of $x$ and $x'$ produced from the hash function $g$ (see \cref{def:hash-stream}). Then with the fixed randomness, $\cS_{i,g}$ and $\cS'_{i,g}$ are neighboring. To see this, observe that for all updates except those inserting or deleting $u$, $\cS_{i,g}$ and $\cS'_{i,g}$ are exactly the same. For updates regarding item $u$,  if $u$ is hashed into substream $i$, then $\cS_{i,g}$ and $\cS'_{i,g}$ will differ only in time $t^*$. Otherwise $\cS_{i,g}$ and $\cS'_{i,g}$ will be identical. Thus, $\cS_{i,g}$ and $\cS'_{i,g}$ will be neighboring streams for all $i \in [L]$.

Let $\cS_{i,B}$ and $\cS'_{i,B}$ be the substreams of $x$ and $x'$ produced after blocklisting (see~\cref{def:bl-stream}). Under the fixed randomness, the timesteps at which items are first blocklisted are the same for the neighboring streams, except (possibly) for item $u$. 
Since the updates from substreams $\cS_{i,B}$ and $\cS'_{i,B}$ are stored exactly as-is in \textsf{DICT}, and then fed to \textsf{BM-Count-Distinct} as input, 
the input streams to \textsf{BinaryMechanism-CD}$_i$ are indeed neighboring.

By \cref{lem:bm-ec-dp}, the output $\hat{s}_i$ of \textsf{BinaryMechanism-CD}$_i$ is $\beta/L$-approximate $\rho/L$-zCDP. The remainders of the operations in \textsf{COUNTING-DICT}$_i$ -- including the thresholding step to output either the numerical value of $\hat{s}_i$ or TOO-HIGH -- are simply postprocessing on the private outputs $\hat{s}_i$ of \textsf{BinaryMechanism-CD}$_i$, which will retain the same privacy guarantee.
Thus, \textsf{COUNTING-DICT}$_i$ is $\beta/L$-approximate $\rho/L$-zCDP as well. 
\end{proof}

\subsubsection{Helper Lemmas}

\begin{lemma}\label{lem:bm-ec-dp}
   The \textsf{BinaryMechanism-CD} instance in \textsf{COUNTING-DICT}$_i$ is $\beta/L$-approximate $\rho/L$-zCDP. 
\end{lemma}
\begin{proof}
Consider the binary tree produced by \textsf{BM-Count-Distinct}$_i$ with $\log(T)$ levels. We define a vector $G_h$ of length $T/2^h$ for each level $h \in [\log(T)]$ of the binary tree as 
$$G_h[j] = s_i[j \cdot 2^h] -  s_i[(j-1) \cdot 2^h] $$
for all $j \in [T/2^h]$ and $s_i[t]=\sum_{u \in \cU} \mathbf{1}_{\textsf{DICT}_i[u][t]>0} $ as defined in \cref{line:s-ec} of \textsf{COUNTING-DICT}$_i$. Let $G=(G_0,\ldots, G_{\log(T)})$. To prove the claim, we will bound the sensitivity of the counts stored in the binary tree represented by $G$, and then show that sufficient noise is added to each count to satisfy differential privacy. Similar to the original binary tree mechanism of \cite{ChanSS2011,DworkNPR10}, the output of \textsf{BM-Count-Distinct}$_i$ at timestep $t$ can be obtained from $G$ by considering the dyadic decomposition of the interval $(0,t]$ as a sum of the individual nodes composing the interval, and this output will be private by postprocessing.\footnote{\cite{JainKRSS23}  used similar techniques to argue about the sensitivity of their binary tree mechanism. However, their argument is more straightforward as it does not have to consider the randomness from hashing or blocklisting.}

First, we claim that the binary tree described by $G$ plus DP noise (which we will determine) produces the output $\hat{s}_i$. To see this, first observe that $\hat{s_i}[t] = s_i[t] - s_i[t-1] +Z[t]$ where $Z[t]$ is the noise term. Also $G_0[j] = \hat{s}_i[j]-Z[j]$. The claim follows by induction over $h \in [\log(T)]$. 

Let $G$ and $G'$ be the binary tree representation of neighboring streams $x$ and $x'$ respectively. We will show that $\| G-G'\|_2 \leq 2\sqrt{(W+1) (\log(T)+1)}$ with probability $1-\beta/L$. 

Fix $h \in [\log(T)]$ and $j \in [T/2^h]$. For ease of notation, let $j_1 = (j-1)\cdot 2^h$ and $j_2 = j \cdot 2^h$. Then 
\begin{align}
    \vert G_h [j] - G'_h[j] \vert = \vert s_i[j_2] - s_i[j_1] - s'_i[j_2] +s'_i[j_1]\vert \leq 2,
\end{align}
where the inequality is due to the fact that $s_i$ and $s'_i$ can differ by at most $1$ at both timesteps $j_1$ and $j_2$.

Next, observe that for a fixed $h$, the intervals $(j_1,j_2]$ are disjoint, by definition. Also, for $j \in [T/2^h]$, $G_h [j] \neq G'_h [j]$ are different in at most $W+1$ intervals with probability $1-\beta$ (by \cref{lem:bl-id}) where $W=T^{2/3}$. 

Thus, with probability $1-\beta/L$, the (squared) $\ell_2$-sensitivity of $G$ is bounded:
\begin{align} \label{eq:l2-sen}
    \Delta^2_2 \leq \|G-G'\|^2_2 = \sum_{h \in [\log(T)] }\sum_{j \in [T/2^h]}(G_h[j] - G'_h[j])^2 \leq (\log T+1)(W+1)\cdot 2^2 
\end{align}

By \cref{thm:gaussian}, adding Gaussian noise sampled $\cN(0,\sigma^2)$ for $\sigma^2 = \frac{\Delta^2_2 L}{2\rho}$ to each count stored in a node of the binary tree represented by $G$ will satisify $\rho/L$-zCDP with probability $1-\beta/L$. Plugging in the bound on $\Delta^2_2$, it is sufficient to add Gaussian noise with variance $\sigma^2 = \frac{(\log T+1)(W+1)\cdot L}{\rho}$. In Algorithm \ref{alg:exact-count}, the \textsf{BinaryMechanism-CD} subroutine is instantiated with privacy parameter $\rho/L$, which adjusts for the extra factor of $L$.

Finally, since the output of \textsf{BinaryMechanism-CD} can be obtained by postprocessing the noisy nodes of $G$, the output is $\rho/L$-zCDP with probability $1-\beta/L$.
\end{proof}

\begin{lemma}\label{lem:privacy_occ_bound}
Suppose $ob=false$. With probability at least $1-\beta/L$, the maximum \occur\ of the stream produced from the blocklisting procedure (\cref{def:bl-stream}) is bounded by $T^{2/3}$.
\end{lemma}

\begin{proof}
Recall that the probability of blocklisting any element after an appearance in the stream is $p = \frac{\log (T^{1/3}L/\beta)}{T^{2/3}}$. 

For any element $x \in \cU$, we can bound the failure probability of the blocklist to catch an element after the maximum number of occurrences:
\begin{align*}
    \Pr[x \notin \cB \text{ after } T^{2/3}  \text{ appearances} ] & = (1 - \frac{\log (T^{1/3}L/\beta)}{T^{2/3}})^{T^{2/3}} \\
   & \leq  e^{-T^{2/3}\cdot \frac{\log(T^{1/3}L/\beta)}{T^{2/3}}}   \\
   & = e^{-\log(T^{1/3}L/\beta)} \\
   & = \frac{\beta}{T^{1/3}L}
\end{align*}
The inequality in the second step comes from the fact that $(1-a) \leq e^{-a}$ for all $a \in \mathbb{R}$. 

At most $T/T^{2/3} = T^{1/3}$ elements can appear $\geq T^{2/3} $ times in a stream of length $T$. Taking a union bound, the probability that \textit{any} of these elements is not blocklisted after $T^{2/3}$ appearances is at most $\frac{\beta}{T^{1/3}L} \cdot {T^{1/3}} = \beta/L $.
\end{proof}

\begin{lemma}\label{lem:bl-id}
    Suppose $ob=false$. Let $x$ and $x'$ be neighboring input streams and fix the randomness of \textsf{CountDistinct'} across runs on $x$ and $x'$. For any $i \in [L]$, 
    $\cS_{i,B}$ and $\cS'_{i,B}$ differ in at most $T^{2/3}+1$ positions with probability $1-\beta/L$. 
\end{lemma}
\begin{proof}
Let neighboring streams $x$ and $x'$ differ only at timestep $t^*$, in which the update (either deletion or addition) in $x$ is for item $u$, and in $x'$ is $\bot$, and fix the randomness used in \textsf{CountDistinct}' across runs on $x$ and $x'$. As shown in the proof of \cref{lem:exact_count_dp}, the substreams $\cS_{i,g}$ and $\cS'_{i,g}$ are neighboring  and thus will also differ at timestep $t^*$ with respect to item $u$. Thus under the fixed randomness, for all timesteps $t \neq t^*$, the same items are blocklisted in $\cS_i$ and $\cS'_i$.
 
    Let $\textbf{E}_0$ be the bad event that there exists an item $v$ that appears in the hashed substream at some timestep $t \neq t^*$ and is \emph{not} blocklisted after $T^{2/3}$ occurrences. 
    By \cref{lem:privacy_occ_bound}, the probability of $\textbf{E}_0$ is bounded by $\beta/L$. 
    We will condition on the event $\textbf{E}_0$ not occurring for the remainder of the proof. Note that if the items in  $\cS_{i,g}$ and $\cS'_{i,g}$ do not appear more than $T^{2/3}$ times for all $i \in [L]$, then naturally $\textbf{E}_0$ does not occur.

    Recall that the item $u$ appears in exactly the same timesteps in $\cS_{i,g}$ and $\cS'_{i,g}$ for $t \neq t^*$. Suppose that  the number of appearances of $u$ in those steps is $\geq T^{2/3}$. 
    The resulting blocklisted streams $\cS_{i,B}$ and $\cS'_{i,B}$ can differ in at most $T^{2/3}+1$ timesteps because the item $u$ may be blocklisted before it appears $T^{2/3}$ times, but conditioned on $\textbf{E}_0$ not occurring, it must be blocklisted after the $T^{2/3}$-th appearance. Since $\cS_{i,g}$ has an extra occurrence of $u$ (at timestep $t^*$) relative to $\cS'_{i,g}$, this means that $\cS'_{i,B}$ can differ from $\cS'_{i,B}$ in at most $T^{2/3}+1$ timesteps, after which both $\cS_{i,B}$ and $\cS'_{i,B}$ will have 0's for all future occurrences of item $u$. 
\end{proof}


\subsection{Proof of Theorem \ref{thm:acc_exact} (Accuracy) }\label{app:exactacc}

We restate Theorem \ref{thm:acc_exact} below for convenience.

\accuracy*

\begin{proof}[Proof of Theorem \ref{thm:acc_exact}.] 
\textbf{When $ob$ is true:}
The proof relies on the following lemmas, which ensure that for a specific timestep $t$, the good events occur with high probability:

\begin{enumerate}
    \item In all substreams $i \in [L]$, the correct number of distinct elements in the substream $i$ by hashing, denoted $F_i(t)$, 
    is also a good estimator for the number of distinct elements in the entire stream at timestep $t$, denoted $F(t)$ (Lemma~\ref{lem:substream_concentration}). That is, for all $i \in [L]$, the following two conditions hold at the same time for any specific timestep $t$  with probability at least $1-\beta/5$ for any $\eta \in (0,0.5)$:

\begin{enumerate}
     \item  $\forall i \in [L]$ with $F(t) \geq 2^i \cdot  \frac{4\lambda}{\eta^2}$, we have $(1-\eta)\frac{F(t)}{2^i} \leq F_i(t) \leq (1+\eta)\frac{F(t)}{2^i}$ 
        \item$\forall i \in [L]$ with $F(t) < 2^i \cdot  \frac{4\lambda}{\eta^2}$, we have $\frac{F(t)}{2^i} - \frac{4\lambda}{\eta}\leq F_i(t) \leq \frac{F(t)}{2^i} + \frac{4\lambda}{\eta}$.  
\end{enumerate}
    
    \item \textsf{BinaryMechanism-CD} (Algorithm~\ref{alg:bm}) is accurate (Lemma~\ref{lem:bm_acc}). That is, for all $i \in [L]$, we have $|F_i(t) - \hat{s}_i (t)| \leq \gamma =   \sqrt{ \frac{  4  (W+1) (\log T + 1)^3  \log (10 (\log T + 1)/ \beta) }{\rho} } $ with probability $1-\beta/5$.

\item For any stream $i$, if the correct number of distinct elements in the subtream $i$ is below a certain threshold then \textsf{COUNTING-DICT}$_i$ will not output TOO-HIGH  (Lemma~\ref{lem:e2-occur}). Plugging in $\beta/5L$ into Lemma~\ref{lem:e2-occur} yields that if $F_i(t) \leq 16 \max \{  \gamma/\eta, 32\lambda / \eta^2 \}$, then \textsf{COUNTING-DICT}$_i$ will not output TOO-HIGH, i.e. the noisy count $\hat{s}_i \leq \tau = 16 \max \{  \gamma/\eta, 32\lambda / \eta^2 \} +   2\sqrt{2} \frac{(\log T + 1)^{3/2}\sqrt{W \log(20T \lceil \log T \rceil /\beta)}}{\sqrt{\rho}}$, with probability at least $1-\beta/5$.  
\end{enumerate}

To prove the desired accuracy claim, we will condition on all three high-probability events listed above occurring at timestep $t$. Note that each of the three events occur with probability $1-\beta/5$. Thus all three events will happen with probability at least $1-\frac35 \beta \geq 1-\beta$ by a union bound.

We consider two cases for the number of distinct elements of the stream at time $t$ denoted by $F(t)$: (1) $F(t) \geq 8 \max (\gamma / \eta, 32 \lambda / \eta^2)$, for which we show that the resulting approximation satisfies a multiplicative error of $(1\pm 4\eta)$, and (2) $F(t) \leq 8 \max (\gamma / \eta, 32 \lambda / \eta^2)$, for which we show that the resulting approximation has an additive error of $32 \max (\gamma / \eta, 32 \lambda / \eta^2)$. We now separately consider the two cases. 

\textbf{Case (1).} $F(t) \geq 8 \max (\gamma / \eta, 32 \lambda / \eta^2)$.  Let $i^* \in [L]$ be the largest $i$ s.t. $\frac{F(t)}{2^{\istar}} \geq 4 \max (\gamma / \eta, 32 \lambda / \eta^2)$.  
Note that by Lemma~\ref{lem:substream_concentration}, $(1-\eta) \frac{F(t)}{2^{i*}} \leq F_{i^*} (t) \leq (1+\eta) \frac{F(t)}{2^{i^*}}$. Therefore by the definition of $i^*$, $F_{i^*}(t) \leq 2 \frac{F(t)}{2^{i^*}} \leq 2 \cdot 8  \max (\gamma / \eta, 32 \lambda / \eta^2) = 16 \max (\gamma / \eta, 32 \lambda / \eta^2)$. Since in this case the noisy count $\hat{s}_{i^*}$ at timestep $t$ would not exceed $\tau$ by Lemma~\ref{lem:e2-occur}, so \textsf{COUNTING-DICT}$_{i^*}$ will not output TOO-HIGH in the stream $i^*$. Then $\cS_{i^*}[t] = \hat{s}_{i^*} (t) $), and by Lemma~\ref{lem:bm_acc}, \textsf{BinaryMechanism-CD} will be accurate and $\cS_{i^*}[t]  \geq F_{\istar}(t) - \gamma$. 
    Since $F_{\istar} (t) \geq (1-\eta) \frac{F(t)}{2^{i*}} \geq 2 \max (\gamma / \eta, 32 \lambda / \eta^2)$ by Lemma~\ref{lem:substream_concentration} and using the fact that $\eta < 0.5$: 
    \begin{equation}\label{eq.sbound}
        \cS_{i^*}[t]  \geq 2 \max (\gamma / \eta, 32 \lambda / \eta^2) - \gamma \geq \max (\gamma / \eta, 32 \lambda / \eta^2).
    \end{equation}
    
    Above we showed that $F_{i^*}(t) \leq 16 \max (\gamma / \eta, 32 \lambda / \eta^2)$. Then by Lemma~\ref{lem:e2-occur}, the noisy count from $i^*$ will not exceed $\tau = 16 \max \{  \gamma/\eta, 32\lambda / \eta^2 \} +   2\sqrt{2} \frac{(\log T + 1)^{3/2}\sqrt{W \log(20T \lceil \log T \rceil /\beta)}}{\sqrt{\rho}}$, 
    so \textsf{COUNTING-DICT}$_{i^*}$ will not output TOO-HIGH. 
    The only concern now is that the output from $i^*$ may not be output if the noisy count is smaller than $\max (\gamma / \eta, 32 \lambda / \eta^2)$ (see Line 22 of Algorithm~\ref{alg:large-universe-sampling}), but by Inequality \eqref{eq.sbound}, this is impossible because $\cS_{i^*}[t]  \geq \max (\gamma / \eta, 32 \lambda / \eta^2)$. Therefore, Algorithm~\ref{alg:large-universe-sampling} with \textsf{COUNTING-DICT}$_{i}$ as the subroutine will produce a non-zero output, i.e. it will output some $\cS_{i'}[t] \cdot 2^{i'}$ for some $i' \geq i^*$ instead of 0.

We now proceed in two steps: first, we derive a lower bound on the true count \( F_{i'}(t) \) in substream \( i' \); then, we bound the ratio between Algorithm~\ref{alg:large-universe-sampling}'s output \(  \cS_{i'}[t]\cdot 2^{i'} \) and the true count \( F(t) \), using \( F_{i'}(t) \) as an intermediate quantity.

We start with the lower bound on $F_{i'}(t)$:
    \begin{align*}
        F_{i'}(t) & \geq \cS_{i'}[t] - \gamma & \text{by Lemma~\ref{lem:bm_acc}}\\
        & \geq (1 - \eta ) \cS_{i'}[t] & \text{because } \gamma \leq \eta \cS_{i'}[t]\\
        & \geq 16 \lambda / \eta^2 & \text{because } 32 \lambda \leq \eta^2 \cS_{i'}[t]  \text{ and } 0 < \eta < 0.5
    \end{align*}

Next we bound the ratio between $\cS_{i'}[t]\cdot 2^{i'}$ and $F(t)$. According to Lemma~\ref{lem:substream_concentration}, $(1-\eta)\frac{F(t)}{2^{i'}} \leq F_{i'}(t) \leq (1+\eta)\frac{F(t)}{2^{i'}}$. Then,
\begin{align*}
    \cS_{i'}[t]  & \leq  F_{i'} (t) + \gamma & \text{by Lemma~\ref{lem:bm_acc}}\\
   \Longrightarrow \qquad \cS_{i'}[t] & \leq \frac{F_{i'} (t)}{ 1-\eta} & \text{because } \gamma \leq \eta \cS_{i'}[t] \\
    & \leq \frac{1+\eta}{1-\eta} \frac{F(t)}{2^{i'}} & \text{ by Lemma~\ref{lem:substream_concentration}} \\ 
    & \leq (1+4\eta) \frac{F(t)}{2^{i'}} & \text{because } 0 < \eta < 0.5 
\end{align*}
To see the second inequality, we use the fact that $\gamma \leq \eta \cS_{i'}[t]$, and plug this into $\cS_{i'}[t] \leq  F_{i'} (t) + \gamma$ to get $\cS_{i'}[t]  \leq  F_{i'} (t) + \eta \cS_{i'}[t]$. Rearranging gives that $\cS_{i'}[t] \leq \frac{F_{i'} (t) }{1-\eta}$. 
The third inequality is from Lemma~\ref{lem:substream_concentration}, which gives that $F_{i'}(t) \leq (1+\eta)\frac{F(t)}{2^{i'}}$ whenever $F(t) \geq 2^{i'} \cdot  \frac{4\lambda}{\eta^2}$. We now show that because we are in Case 1 where $F(t) \geq 8 \max (\gamma / \eta, 32 \lambda / \eta^2)$, then it must always be the case that $F(t)/2^{i'} \geq  \frac{4\lambda}{\eta^2}$. Assume towards a contradiction that $F(t)/2^{i'} <  4\lambda / \eta^2$. Then by Lemma~\ref{lem:substream_concentration} $F_{i'}(t) \leq F(t)/2^{i'} + 4\lambda / \eta$. By the lower bound above $F_{i'}(t) \geq 16\lambda / \eta^2$. Combining these gives that,
\[
16\lambda / \eta^2 \leq F_{i'}(t) \leq F(t)/2^{i'} + 4\lambda / \eta < 4\lambda / \eta^2 + 4\lambda / \eta \leq 8\lambda / \eta^2,
\]
where the second to last step is from the assumption that $F(t)/2^{i'} <  4\lambda / \eta^2$ and the last step is because $\eta \in (0, 0.5)$. Clearly this is a contraction, so it must be that $F(t)/2^{i'} \geq  \frac{4\lambda}{\eta^2}$.

By symmetric arguments, 
\[
    \cS_{i'}[t]   \geq F_{i'}(t) - \gamma \geq \frac{F_{i'} (t)}{1+\eta}  \geq \frac{1-\eta}{1+\eta} \frac{F(t)}{2^{i'}} \geq (1-4\eta ) \frac{F(t)}{2^{i'}}.
\]
Thus in the case when $F(t) > 8 \max (\gamma / \eta, 32 \lambda / \eta^2)$, \cref{alg:large-universe-sampling} produces a numerical output $\cS_{i'}[t]\cdot 2^{i'}$ for some $i'>0$, which will achieve a multiplicative error of $1 \pm 4 \eta$ with respect to the true count $F(t)$.

\textbf{Case (2).} $F(t) \leq 8 \max (\gamma / \eta, 32 \lambda / \eta^2)$, in which case again $F_i(t) \leq 16  \max (\gamma / \eta, 32 \lambda / \eta^2)$ for all $i$, so the noisy count would not exceed $\tau$ by Lemma~\ref{lem:e2-occur}, so \textsf{COUNTING-DICT}$_{i}$ will not output TOO-HIGH. Then, Algorithm~\ref{alg:large-universe-sampling} either outputs 0, which will result in additive error at most $8  \max (\gamma / \eta, 32 \lambda / \eta^2) $), or it outputs $\cS_{i'}[t] \cdot 2^{i'}$ for some $i'$, such that $\cS_{i'}[t]  \geq \max (\gamma / \eta, 32 \lambda / \eta^2)$.  In this latter case, by the same argument as in Case (1),
\begin{align*}
    F_{i'}(t)  \geq \cS_{i'}[t] - \gamma \geq (1-\eta) \cS_{i'}[t] \geq 16\lambda / \eta^2.
\end{align*}

 Having established a lower bound for $F_{i'}(t)$, we now turn to derive an upper bound for the algorithm's output $\cS_{i'}[t] \cdot 2^{i'}$. 
By Lemma~\ref{lem:substream_concentration}, we know that if $\frac{F(t)}{2^{i'}} \geq \frac{4\lambda}{\eta^2}$ then $F_{i'}(t) \leq 2 \frac{F(t)}{2^{i'}}$, and otherwise, if $\frac{F(t)}{2^{i'}} < \frac{4\lambda}{\eta^2}$, then $16 \lambda / \eta \leq 16\lambda / \eta^2 \leq F_{i'}(t) \leq \frac{F(t)}{2^{i'}} + \frac{4\lambda}{\eta}$ which implies that $\frac{F(t)}{2^{i'}} \geq 12 \lambda / \eta$. Thus in this case $F_{i'}(t) \leq \frac{F(t)}{2^{i'}} + 4\lambda / \eta \leq 2\frac{F(t)}{2^{i'}}$ as well, since $F(t) / 2^{i'} \geq 12\lambda / \eta$.
Hence under both conditions, 
\begin{align*}
    \cS_{i'}[t]  \cdot 2^{i'} & \leq   (F_{i'} (t) + \gamma) \cdot 2^{i'} & \text{by Lemma~\ref{lem:bm_acc}}\\
    & \leq \frac{F_{i'}(t)}{1-\eta} \cdot 2^{i'} &\text{because } \gamma \leq \eta \cS_{i'}[t] \\
    & \leq 2 \cdot 2 \cdot \frac{F(t)}{2^{i'}} \cdot 2^{i'} & \text{ because } F_{i'}(t) \leq 2 \frac{F(t)}{2^{i'}} \text{ and } 0 < \eta < 1/2\\
    & = 4F(t) &\\
    & \leq 32 \max (\gamma / \eta, 32 \lambda / \eta^2)& \text{by Case 2 condition}
\end{align*}

Therefore, in this case Algorithm~\ref{alg:large-universe-sampling} achieves an additive error of at most $32  \max (\gamma / \eta, 32 \lambda / \eta^2)$.

\paragraph{When $ob$ is false:} 
This proof is very similar to the case where $ob$ is true. The main difference is that we need one extra lemma about the size of the blocklist. There is also an additional error caused by the blocklist, which can be treated as part of the error from the binary mechanism. The majority of the analysis stays the same; the only difference is that we will have a larger error in this case because of the blocklist. 

We will use the three same key lemmas as in the case where $ob$=true, and an additional lemma bounding the size of the blocklist. These four lemmas will ensure that at a fixed timestep $t$, the desired good events will happen with high probability:

\begin{enumerate}
   \item In all substreams $i \in [L]$, the correct number of distinct elements in the substream $i$ by hashing, denoted $F_i(t)$, 
    is also a good estimator for the number of distinct elements in the entire stream at timestep $t$, denoted $F(t)$ (Lemma~\ref{lem:substream_concentration}). That is, for all $i \in [L]$, the following two conditions hold at the same time for any specific timestep $t$  with probability at least $1-\beta/5$ for any $\eta \in (0,0.5)$:

\begin{enumerate}
     \item  $\forall i \in [L]$ with $F(t) \geq 2^i \cdot  \frac{4\lambda}{\eta^2}$, we have $(1-\eta)\frac{F(t)}{2^i} \leq F_i(t) \leq (1+\eta)\frac{F(t)}{2^i}$ 
        \item$\forall i \in [L]$ with $F(t) < 2^i \cdot  \frac{4\lambda}{\eta^2}$, we have $\frac{F(t)}{2^i} - \frac{4\lambda}{\eta}\leq F_i(t) \leq \frac{F(t)}{2^i} + \frac{4\lambda}{\eta}$.  
\end{enumerate}
    
    \item \textsf{BinaryMechanism-CD} (Algorithm~\ref{alg:bm}) is accurate (Lemma~\ref{lem:bm_acc}). That is,         for all $i \in [L]$, we have $|F_i(t) - \hat{s}_i (t)| \leq \gamma =   \sqrt{ \frac{  4  (W+1) (\log T + 1)^3  \log (10 (\log T + 1)/ \beta) }{\rho} } $ with probability at least $1-\beta/5$.

\item For any stream $i$, if the correct number of distinct elements in the subtream $i$ is below a certain threshold then \textsf{COUNTING-DICT}$_i$ will not output TOO-HIGH  (Lemma~\ref{lem:e2-occur}). We plug $\beta/5L$ into Lemma~\ref{lem:e2-occur} to obtain that, if $F_i(t) \leq 16 \max \{  \gamma/\eta, 32\lambda / \eta^2 \}$, then \textsf{COUNTING-DICT}$_i$ will not output TOO-HIGH, i.e. the noisy count $\hat{s}_i \leq \tau = 16 \max \{  \gamma/\eta, 32\lambda / \eta^2 \} +   2\sqrt{2} \frac{(\log T + 1)^{3/2}\sqrt{W \log(20T \lceil \log T \rceil /\beta)}}{\sqrt{\rho}}.$ 

\item The blocklist has bounded size (Lemma~\ref{lem:blocklist_size}).
With probability at least $1 - \beta / 5$, the size of the blocklist is bounded by $3T^{1/3} \log (T^{1/3} \lceil \log T \rceil  / \beta)$.        
\end{enumerate}

Conditioned on these good events and plugging in $W=T^{2/3}$, by Lemmas \ref{lem:bm_acc} and~\ref{lem:blocklist_size}, for each instance of \textsf{BinaryMechanism-CD}, the overall additive error stemming from the binary mechanism and blocklisting is $\gamma=  \sqrt{ \frac{  4  (T^{2/3}+1) (\log T + 1)^3  \log (10 (\log T + 1)/ \beta) }{\rho} } + 3T^{1/3} \log (T^{1/3}  \lceil \log T \rceil  / \beta)$. 

To prove the desired accuracy claim, we will condition on all four high-probability events described above occurring at timestep $t$. Note that each of the four events occur with probability $1-\beta/5$. Thus all four events will happen with probability at least $1-\frac45 \beta \geq 1-\beta$ by a union bound.

From here, we can again treat separately two cases based on $F(t)$, the number of distinct elements of the stream at time $t$: (1) $F(t) \geq 8 \max (\gamma / \eta, 32 \lambda / \eta^2)$ and (2) $F(t) \leq 8 \max (\gamma / \eta, 32 \lambda / \eta^2)$. As in the case where $ob$=true, we show in Case (1), the resulting approximation satisfies a multiplicative error of $(1\pm 4\eta)$, and in Case (2), the resulting approximation has an additive error of $32 \max (\gamma / \eta, 32 \lambda / \eta^2)$. The analysis in both cases is identical to when $ob$=true as presented above, with the correspondingly larger value of $\gamma$, and so is not repeated here. 
\end{proof}

\subsection{Proofs of Helper Lemmas for Theorem \ref{thm:acc_exact}}\label{app.helperlemmas}

Lemma \ref{lem:substream_concentration} bounds the number of elements in the substream after hashing. Lemma \ref{lem:bm_acc} proves the accuracy of \textsf{BinaryMechanism-CD} algorithm. Lemma \ref{lem:blocklist_size} bounds the size of the blocklist when $ob$ is false. With the help of these lemmas, we can show the accuracy of \textsf{COUNTING-DICT}, as an intermediate step in the analysis.

\begin{lemma}[Substream concentration bound]\label{lem:substream_concentration}
    Let $F(t)$ be the number of distinct elements of the stream at time $t$, and let $F_i(t)$ be the number of distinct elements of substream $\cS_{i,g}$ (\cref{def:hash-stream}) for any $i \in [L]$. If $F(t) \geq 2^i \cdot  \frac{4\lambda}{\eta^2}$, then $\Pr[|F_i(t) -  \frac{F(t)}{2^i}| > \eta \cdot \frac{F(t)}{2^i} ] \leq \frac{\beta}{5L}$. Otherwise, if  $F(t) \leq 2^i \cdot  \frac{4\lambda}{\eta^2}$, then $\Pr[|F_i(t) -  \frac{F(t)}{2^i}| > \frac{4\lambda}{\eta} ] \leq \frac{\beta}{5L}$. This implies that the following two conditions hold at the same time with probability at least $1-\beta/5$:
\begin{enumerate}
     \item  $\forall i \in [L]$ with $F(t) \geq 2^i \cdot  \frac{4\lambda}{\eta^2}$, we have $(1-\eta)\frac{F(t)}{2^i} \leq F_i(t) \leq (1+\eta)\frac{F(t)}{2^i}$ 
        \item$\forall i \in [L]$ with $F(t) < 2^i \cdot  \frac{4\lambda}{\eta^2}$, we have $\frac{F(t)}{2^i} - \frac{4\lambda}{\eta}\leq F_i(t) \leq \frac{F(t)}{2^i} + \frac{4\lambda}{\eta}$.  
\end{enumerate}
\end{lemma}

\begin{proof} 
We start with the case of $F(t) \geq 2^i \cdot  \frac{4\lambda}{\eta^2}$. 
Applying Lemma~\ref{lem:hash} to the $F_i(t)$ as a sum of $F(t)$ $\lambda$-wise independent $Bernoulli(2^{-i})$ random variables, and with $\mu = \frac{F(t)}{2^{i}}, A=\eta \cdot \frac{F(t)}{2^i}$ yields the following:
\begin{align*}
\Pr\left[|F_i(t) -  \frac{F(t)}{2^i}| > \eta \cdot \frac{F(t)}{2^i} \right] < \  & 8\left( \frac{\lambda \cdot \frac{F(t)}{2^i} + \lambda^2}{\eta^2 \cdot (\frac{F(t)}{2^i})^2} \right)^{\lambda/2} \\
= \ & 8\left( \frac{\lambda}{\eta^2 \frac{F(t)}{2^i}} + \frac{\lambda^2}{\eta^2 (\frac{F(t)}{2^i})^2} \right)^{\lambda/2}  \\
\leq \ & 8\left( \frac{\lambda}{\eta^2 (4\lambda / \eta^2)} + \frac{\lambda^2}{\eta^2 (4\lambda/\eta^2)^2} \right)^{\lambda/2}  \\
 = \ &  8 \left(\frac{1}{4} + \frac{\eta^2}{16}\right)^{\lambda/2},
    \end{align*}
where the third step is because of the case $F(t) \geq 2^i \cdot  \frac{4\lambda}{\eta^2}$, and the fourth step simplifies terms. 

We now wish to bound the right hand side by $\frac{\beta}{5L}$; solving this inequality for $\lambda$ yields the following bound:
    \begin{align*}
& 8 \left(\frac{1}{4} + \frac{\eta^2}{16}\right)^{\lambda/2}  \leq \frac{\beta}{5L} \\
\iff \quad  & \left(\frac{1}{4} + \frac{\eta^2}{16}\right)^{\lambda/2} \leq \beta/40L \\
        \iff \quad & \frac{\lambda}{2} \log (1/4 + \eta^2 / 16) \leq \log (\beta/ 40L) \\
        \iff \quad & \lambda \geq \frac{2\log (\beta / 40L)}{\log (1/4 + \eta^2 / 16)} \\
        \iff \quad & \lambda \geq \frac{2\log (40L/\beta)}{\log (\frac{1}{1/4 + \eta^2 / 16})} 
    \end{align*}
    The last step comes from multiplying both the numerator and denominator by $-1$, and $-\log(x) = \log (1/x)$.

    Since $\eta < 0.5$, then the denominator can be bounded by: $\log (\frac{1}{1/4 + \eta^2 / 16}) > \log (64/17) > 1$. 
    Since Algorithm \ref{alg:large-universe-sampling} sets $\lambda = 2 \log (40L / \beta)$, the above inequality will be satisfied. 

    Next we consider the second case, where $F(t) < 2^i \cdot  \frac{4\lambda}{\eta^2}.$  Applying Lemma~\ref{lem:hash} again to the $F_i(t)$, now with $A= \frac{4 \lambda}{\eta}$, yields the following:
    \begin{align*}
     \Pr\left[|F_i(t) -  \frac{F(t)}{2^i}| >  \frac{4\lambda}{\eta} \right] 
& < \   8\left( \frac{\lambda \cdot \frac{F(t)}{2^i} + \lambda^2}{(\frac{4\lambda}{\eta})^2} \right)^{\lambda/2} \\
< \ &  8\left( \frac{\lambda \cdot 4\lambda / \eta^2 + \lambda^2}{ 16\lambda^2 / \eta^2} \right)^{\lambda/2} \\
& = 8 \left(\frac{1}{4} + \frac{\eta^2}{16}\right)^{\lambda/2},
    \end{align*}
where the second step is because of the case $F(t) < 2^i \cdot  \frac{4\lambda}{\eta^2}$, and the third step simplifies terms. 

   We again wish to bound the right hand side by $\frac{\beta}{5L}$, and by the same steps as in the first case, the desired inequality holds if and only if $\lambda \geq \frac{2\log (40L/\beta)}{\log (\frac{1}{1/4 + \eta^2 / 16})}$, and using the requirement that $\eta < 0.5$, Algorithm \ref{alg:large-universe-sampling}'s choice of $\lambda = 2 \log (40L / \beta)$ will satisfy the desired inequality.
\end{proof}

\begin{lemma}[Binary mechanism accuracy] \label{lem:bm_acc} Fix a timestep $t \in [T]$, and recall that $\hat{s}_i(t)$ is the noisy count of $F_i(t)$ produced by \textsf{BinaryMechanism-CD}$_i$.
Then $|F_i(t) - \hat{s}_i(t)| \leq  \sqrt{ \frac{  4  (W+1) (\log T + 1)^3  \log (10 (\log T + 1)/ \beta) }{\rho} }$ simultaneously for all $i \in [L]$ with probability $1-\beta/5$.
\end{lemma}

\begin{proof}
    We will show that for a specific substream $i \in [L]$, the error $|F_i(t) - \hat{s}_i(t)|$ is bounded by the desired term with probability $1-\beta / 5L$, and the lemma will follow by a union bound over all $i \in [L]$. 
    
    Now consider the error of substream $i$ at timestep $t$. Algorithm~\ref{alg:bm} (\textsf{BinaryMechanism-CD}) adds a total of $\textsf{Bin}_1(t)$ independent Gaussian noise terms, where $\textsf{Bin}_1(t)$ is the number of ones in the binary representation of $t$. 
    Hence, the error $|F_i(t) - \hat{s}_i(t)|$ is a sum of at most $\log T + 1$ independent Gaussian random variables each distributed as $\mathcal{N}(0, 2L(W+1) (\log T + 1) / \rho)$, and so the error itself is also Gaussian with mean 0 and variance at most $\frac{2L(W+1) (\log T + 1) ^2}{\rho }$. 
    
    Applying a Chernoff bound (Lemma~\ref{lem:chernoff}) to this random variable, gives that for any $\gamma>0$,
    \begin{align*}
        \Pr \left[ |F_i(t) - \hat{s}_i(t)| > \gamma  \right] \leq 2 \exp \left( \frac{- \rho \gamma^2}{4L (W+1) (\log T + 1)^2 } \right).
    \end{align*}
    We wish to bound this probability by $\beta / 5L$. Solving this inequality for $\gamma$ yields $\gamma \geq  \sqrt{\frac{  4 L (W+1) (\log T + 1)^2  \log (10L/ \beta) }{\rho}}$. Since $L = \lceil \log T \rceil  \leq \log T + 1$, then choosing $\gamma = \sqrt{ \frac{  4  (W+1) (\log T + 1)^3  \log (10 (\log T + 1)/ \beta) }{\rho} }$ will satisfy the inequality.   
\end{proof}

\begin{lemma} [Bounded blocklist size] \label{lem:blocklist_size}  Suppose $ob = false$. Fix a timestep $t \in [T]$. With
probability at least $1 - \beta / 5$, the size of the blocklist $\cB$ is bounded by $3T^{1/3} \log ( T^{1/3} L/ \beta)$, where $L = \lceil \log T \rceil$. 
    
\end{lemma}

\begin{proof}
    The size of the blocklist is non-decreasing since elements are never removed, so it is sufficient to upper bound the final size of the blocklist after $T$ timesteps. 

    Define the random variable $Y_j$ to be 1 if the $j$-th arrival in the stream is blocklisted and 0 if otherwise. Then $Y = \sum _{i=j}^T Y_j$ is an upper bound on the size of the blocklist (because the arrivals may be updates of the same element). The sampling rate for blocklisting is $p = \frac{\log (T^{1/3} L/\beta)}{T^{2/3}}$, so  $Y_j \sim Bern(\frac{\log (T^{1/3} L / \beta)}{T^{2/3}} ) $ and $\mathbb{E}[Y] =   \frac{T \log (T^{1/3} L / \beta)}{T^{2/3}} = T^{1/3}\log (T^{1/3}L / \beta)  $. Applying a  multiplicative Chernoff bound (Theorem \ref{thm:mult_chernoff}) with $\eta = 2$ yields: 
    $$
\Pr \left[ Y >  3T^{1/3} \log (T^{1/3} L / \beta)  \right] \leq \exp \left( - T^{1/3}\log (T^{1/3} L / \beta) 
 \right) 
$$
We would like to bound this probability by $\beta/5$, which occurs when $T^{1/3} \log(\frac{T^{1/3} L }{\beta}) \geq \log (\frac{5}{\beta})$. This inequality will clearly hold if $T^{1/3} L = T^{1/3} \lceil \log T \rceil \geq 5$, which is true for $T \geq 8$. 
\end{proof}

\section{Omitted Proofs from Section \ref{sec:lb}}

\subsection{Proof of Theorem \ref{thm:lb}}\label{s.appprooflb}

\blocklistlb*

\begin{proof}

We now define a random process defining a stream distribution that is simultaneously hard for the the problem of $\bf{blocklist}_{flip}(W)$ and for $\bf{blocklist}_{occ}(W)$. For notational convenience, set $n = T/2$, and $m = n / W = T/(2W)$. For simplicity, we assume $T,W$ such that $n,m$ are integers.

Let the element universe be $\cU = [n]$, over which we assume there is a total order. We define a distribution of problem instances with the following process:
\begin{itemize}
    \item A uniformly random set $X \subseteq \cU$ of size $m$ is sampled.
    \item Generate the first $n = T/2$ timesteps ensuring that all elements in $X$ have $W$ flippancy and $W$ \occur. 
    \item Then generate the second $n = T/2$ timesteps ensuring that all elements in $\cU$ are inserted once. 
\end{itemize}

Define a stream $x(X)$ to be a deterministic function of $X$ as follows.  For the first half, all updates related to each element in $X$ appear consecutively and in the total order. Each element $u \in X$ has exactly $W$ updates, alternating $W/2$ times between one insertion and one deletion of $u$; this results in both $W$ flippancy and $W$ \occur\ for all $u \in X$ at the end of the first half of the stream. Then, the second half of the stream has one insertion of each element in $\cU$, appearing according to the total order. 

It is easy to see that the correct output for this stream in both the $\bf{blocklist}_{flip}(W)$ and $\bf{blocklist}_{occ}(W)$ problems is to always output 0 in the first half, and to output 1 in the second half only for elements in $X$: $o^*(x(X))_t = 0$ for $t\le n$ and $o^*(x(X))_t = 1$ iff $x_t \in X$ and $t\ge n+1$.

We want to show a space lower bound for an algorithm $\cA_{flip}$ (resp., $\cA_{occ}$) that will satisfy the conditions of no false negatives, and at most $r$ false positives, with probability at least $1-\beta$ over this distribution, for the $\bf{blocklist}_{flip}(W)$ (resp., $\bf{blocklist}_{occ}(W)$) problem. The proof proceeds identically for both problems, so we prove it only for $\cA_{flip}$.

Without loss of generality, we assume $\cA_{flip}$ is deterministic, since for any randomized algorithm, there will exist a random seed to achieve no worse guarantee on the considered input distribution. 

We prove this space lower bound via an  information theory argument (see Section \ref{sec:info} for information theory basics) by considering three random variables in addition to $X$:
\begin{itemize}
\item $S$: A random variable representing the memory state of algorithm $\cA_{flip}$ after observing the first $n=T/2$ timesteps.
\item $Y$: The set of elements with an output 1 in the second half of the stream by algorithm $\cA_{flip}$.
\item $P$: An indicator variable that is 1 if $\cA_{flip}$ has no false negatives and at most $r$ false positives, and 0 otherwise. By definition we know $\Pr[P=1] \geq 1- \beta$.
\end{itemize}

Since the algorithm $\cA_{flip}$ is deterministic, then $S, Y, P$ are all deterministic functions of $X$, since the stream itself is a deterministic function of $X$. Since algorithm $\cA_{flip}$ does not learn new information in the second half of the stream -- since the second half of the stream is the same for all streams -- then  $Y$ is a deterministic function of $S$.

To start, the maximum size of $S$ can be lower bounded by its entropy $H(S)$. 
Then $H(S)$ can be further lower bounded by the mutual information between the input secret $X$ and the random variable $Y$, which is part of algorithm $\cA_{flip}$'s output, conditioned on $P$. This can be seen using inequalities from Fact \ref{fact:infofact}:
\[
 |S| \geq H(S) \geq H(S|P) \geq I(X; S | P) \geq I(X; Y | P).
\]

Next, we apply the definitions of condition mutual information and mutual information:
\begin{align*}
I(X ; Y| P)  &= \Pr[P=1]\cdot I(X ; Y| P=1)  +\Pr[P=0]\cdot I(X ; Y| P=0)  \\
&\geq \Pr[P=1]\cdot I(X ; Y| P=1) \\
&= \Pr[P=1]\cdot \left( H(X|P=1) - H(X|P=1,Y) \right).
\end{align*}

For $H(X|P=1)$, since $\cA_{flip}$ is deterministic and $\Pr[P=1] \geq  1- \beta$, we know that $X | P=1$ is distributed uniformly across at least $(1-\beta) \cdot \binom{n}{m}$ sets of universe elements\footnote{To see this, first observe that $P(X)=1$ partitions the input space according to the event that $\cA_{flip}$ has no false negatives and at most $r$ false positives. Let $R$ denote the size of the support for $X \vert P=1$. Then, since $X$ is uniformly chosen, $\Pr[P(X)=1] \geq 1-\beta$ means that $\frac{R}{\binom{n}{m}} \geq 1-\beta$, and therefore $R \geq (1-\beta)\cdot \binom{n}{m}$ }, and therefore,
\[
H(X|P=1) \geq \log\left( (1-\beta)\cdot \binom{n}{m}\right).
\]

Next we analyze $H(X|P=1,Y)$. Conditioned on $P=1$, by the no false negatives requirement, we know $X \subseteq Y$, and by at most $r$ false positives, we know $|Y| \leq m + r$. Therefore, for any $y$ such that $\Pr[Y=y|P=1] > 0$, it must be that $|y| \leq m+r$ and that conditioned on $P=1$ and $Y=y$, $X$ is a subset of $y$ of size $m$. Therefore $X$ can take at most $\binom{m+r}{m}$ different values. By Fact \ref{fact:infofact}, 
\[
H(X|P=1, Y=y) \leq \log \binom{m+r}{m}
\]
and
\[
H(X|P=1,Y) = \sum_{y \subseteq [n]} \Pr[Y=y|P=1] \cdot H(X|P=1, Y=y) \leq \log \binom{m+r}{m}.
\]

Putting everything together:
\begin{align*}
|S| \geq H(S) \geq I(X ; Y| P) &\geq  \Pr[P=1]\cdot \left( H(X|P=1) - H(X|P=1,Y) \right)\\
&\geq (1-\beta) \cdot \left( \log\left((1-\beta)\cdot \binom{n}{m} \right) - \log \binom{m+r}{m} \right) \\
&= (1-\beta) \cdot \left( \log(1-\beta) + \log\frac{n \times (n-1)\times \cdots \times(n-m+1)}{(m+r) \times (m+r-1) \times \cdots \times (r+1) } \right) \\
&\geq (1-\beta) \cdot \left( \log(1-\beta)  + m \log\frac{n}{m+r}\right)\\
&= (1-\beta) \cdot \left( \log(1-\beta)  + \frac{T}{2W} \log\frac{WT}{T+2Wr}\right). 
\end{align*}
\end{proof}

\subsection{Proof of Corollary \ref{corol:blocklist-ub}}\label{s.appproofub}

\blocklistub*

\begin{proof}
For the space complexity of Algorithm~\ref{alg:large-universe-sampling}, \cref{thm:space} says that when $|\cU| = \text{poly}(T)$ and $ob$=false, then with probability at least $1-\beta$, Algorithm~\ref{alg:large-universe-sampling} uses space $O(T^{1/3} \cdot \text{polylog}(T/\beta)) \cdot \text{poly}(\frac{1}{\rho \eta})$.

For the claim that Algorithm~\ref{alg:large-universe-sampling} has no false negatives with high probability, \cref{lem:privacy_occ_bound} in \cref{app.exactcountproof} shows that when $ob=false$, with probability at least $1-\beta/L$, the maximum \occur\ of the stream produced from the blocklisting is at most $T^{2/3}$. Since Algorithm~\ref{alg:large-universe-sampling} uses $W=T^{2/3}$, then with probability at least $1-\beta/L$, Algorithm~\ref{alg:large-universe-sampling} will have no false negatives, since it will never allow elements with \occur\ larger than $W=T^{2/3}$ to persist without being blocklisted. Plugging in the value of $L=\lceil \log(T) \rceil$, note that $\beta/\lceil \log(T)\rceil \leq \beta/2$, therefore with probability $1-\beta/2$, Algorithm~\ref{alg:large-universe-sampling} will have no false negatives.

 To see the false positive bound, let $X_i$ be an indicator random variable where $X_i=1$ if the element at timestep $i$ is a false positive in Algorithm~\ref{alg:large-universe-sampling}, and let $X=\sum^T_{i=1}X_i$. Recall that Algorithm~\ref{alg:large-universe-sampling} samples elements for the blocklist with probability $p = \frac{\log (T^{1/3} L/\beta)}{T^{2/3}}$ at each occurrence.  
Note that $\Pr[X_i=1] = \Pr[(x_i \in \cB) \cap (x_i<W)] \leq \Pr[(x_i \in \cB) ] = p$, so then $\E[X]= \sum^T_{i=1} \Pr[X_i=1] \leq T\cdot p = T^{1/3} \log (T^{1/3} L/\beta)$. Applying a multiplicative Chernoff bound (Theorem \ref{thm:mult_chernoff}) with $\eta=1$ yields: 
\[\Pr[X \geq 2T^{1/3}{\log (T^{1/3} L/\beta)}]\leq \exp(- (T^{1/3}/3)\log(T^{1/3}L/\beta) ) = \left(\frac{\beta}{T^{1/3}L}\right)^{T^{1/3}/3} \leq \beta/2.\] 
To see the last inequality, we first plug in $L= \lceil \log(T) \rceil$, so we wish to show $\left(\frac{\beta}{T^{1/3}\lceil \log(T) \rceil}\right)^{T^{1/3}/3} \leq \beta/2$. We observe numerically that this holds for $T \geq 8$. 

Taking a union bound over the failure probability from the space complexity, false positive, and false negative bounds, then with probability $1-2\beta$, \cref{alg:large-universe-sampling} solves the \textbf{blocklist}$_{occ}(W)$ for the desired $r$ and space conditions.
\end{proof}

\end{document}